\documentclass[10pt,twocolumn,twoside]{IEEEtran}

\renewcommand{\baselinestretch}{0.97}

\usepackage{multicol, multirow}
\usepackage{array}
\usepackage{amsmath, amsthm, amsfonts, amssymb}
\usepackage{cite}
\usepackage{graphicx}
\usepackage{subfigure}
\usepackage{color}
\usepackage{cases}
\usepackage[]{footmisc}
\usepackage{calc}
\usepackage{physics}
\usepackage{booktabs}
\usepackage{placeins}
\usepackage{xcolor}
\usepackage{hhline}
\usepackage{kotex}
\usepackage{soul}
\usepackage{dblfloatfix}
\usepackage{algorithm,algorithmic}
\usepackage[normalem]{ulem} 

\newcommand{\Conv}{%
  \mathop{\scalebox{1.5}{\raisebox{-0.2ex}{$\circledast$}}
  }
}

\newcommand {\blue} {\color{blue}}

\ifCLASSOPTIONonecolumn
   \renewcommand{\baselinestretch}{2.0}  
   \newcommand{\figwidth}{0.6\columnwidth}
   
   \newcommand{\widefigwidth}{0.5\columnwidth} 
 \else
   \newcommand{\figwidth}{0.95\columnwidth}  
   
   \newcommand{\widefigwidth}{0.98\columnwidth}   
\fi

\newcommand{\ignoreIt}[1] {}

\allowdisplaybreaks 
\newtheorem{theorem}{Theorem}
\newtheorem{corollary}{Corollary}
\newtheorem{lemma}{Lemma}
\newtheorem{prop}{Proposition}
\newtheorem{remark}{Remark}

\DeclareMathOperator*{\argmax}{argmax}

\begin{document}
\author{Hyeonsik~Yeom and Jinyoung~Lee\\
\thanks{
H. Yeom is with the Defense ICT Convergence Research Section, Electronics and Telecommunications Research Institute (ETRI), Daejeon 34129, South Korea, email: hyeonsik@etri.re.kr.

J. Lee is with the Department of Electronics and Electrical Information Engineering, National Korea Maritime \& Ocean University, Busan 49112, South Korea, email: haetsal120@gmail.com.
 }
}

\title{Covert Communication with Spatially Heterogeneous User Cooperation Against a Geometry-Aware Warden}
\ifCLASSOPTIONonecolumn
  \renewcommand{\baselinestretch}{1.5}  
\fi

\maketitle

\begin{abstract}
This paper investigates covert communication in which a single covert user is assisted by multiple spatially distributed non-covert users against a geometry-aware warden. Unlike prior works based on receiver-centric activation, homogeneous channels, or averaged geometry, the proposed framework retains the individual user-to-warden large-scale fading coefficients in both cooperative user activation and the warden’s detection analysis. Under the adopted analytical formulation for large-scale user deployments, we establish an on--off activation structure that balances the interference caused at Bob and the interference effectiveness at warden. Based on the resulting aggregate-interference statistics, we derive a closed-form approximation of the detection error probability, the detection threshold that minimizes it, and the minimum number of cooperative users required to satisfy the covert constraint. Then, we reduce the original high-dimensional power-control problem to a one-dimensional piecewise search and develop a low-complexity algorithm that avoids resolution-dependent grid search. The analysis shows that greater warden-side spatial dispersion can reduce the required cooperative users, whereas a larger average user-to-warden distance increases it. Numerical results validate the analysis and show that finite-sample detection approaches the large-sample benchmark, while imperfect user-to-Bob channel estimation mainly reduces the achievable rate through activation errors.
\end{abstract}


\ifCLASSOPTIONonecolumn
  \clearpage
  \pagenumbering{arabic}
\fi

\IEEEpeerreviewmaketitle
\section{Introduction} 

With the explosive proliferation of wireless devices and the Internet of Things (IoT), modern wireless networks are required to manage enormous volumes of sensitive information. At the same time, the continuous growth in wireless connectivity and device density increases the exposure of communication channels, making the wireless medium highly vulnerable to adversarial surveillance and eavesdropping. {\blue To address these threats, physical layer security (PLS) has emerged as a complementary paradigm to cryptography, exploiting the intrinsic characteristics of wireless channels at the physical layer to provide confidentiality and resilience directly at the signal level \cite{Jiang24Physical, Cao26Reliable, Cao26Self}}. Among the various PLS techniques, \textit{covert communication} has recently attracted considerable attention for its unique objective. Unlike conventional cryptographic approaches that protect the content of a message, covert communication aims to conceal the existence of a transmission from adversarial detectors, typically called a warden. In other words, the transmitter seeks to ensure an extremely low probability that detectors can correctly decide whether a transmission is occurring, while simultaneously guaranteeing a desired data rate for the legitimate receiver. Such concealment capability is essential for diverse mission-critical and privacy-sensitive applications, including military operations, secure IoT networks, and covert ad hoc communications \cite{Jiang24Physical}.

Early covert transmission techniques, including direct sequence spread spectrum, frequency hopping, and chirp spread spectrum, were primarily implementation-driven, motivating theoretical studies on the fundamental limits of covert communication. In \cite[Theorem~1]{Bash13Limits}, Bash \textit{et al.} introduced the \textit{square root law} for additive white Gaussian noise (AWGN) channels. The law states that, under the classical AWGN model without additional uncertainty at a warden, no more than $\mathcal{O}(\sqrt{n})$ bits can be transmitted reliably and covertly within $n$ channel uses. Subsequent works \cite{Che13Reliable,Bloch16Covert,Wang16Fundamental,Arumugam16Keyless,Arumugam19Covert,Tan19Time,Cho21Treating} generalized this principle to binary symmetric, discrete memoryless, multiple-access, broadcast, and interference channels, thus establishing the theoretical foundation of covert communication. 

Under this classical model, the achievable transmission rate approaches zero as the number of channel uses increases. Motivated by this limitation, numerous studies have explored ways to achieve a positive covert rate by introducing uncertainty at a warden. Major approaches include exploiting channel uncertainty~\cite{Lee15Achieving,Shahzad21Covert}, inducing interference via relays or jammers~\cite{Wang19Covert,Sun21Covert,Sun23Covertness,Bai22OnCovert,Sobers17Covert,Li20Optimal}, and utilizing full-duplex receivers~\cite{Shahzad18Achieving,Hu19Covert}. 
Recent studies have also investigated the inherent directionality of mmWave systems \cite{Zhang22Multi}, the spatial reconfigurability of movable antennas \cite{Mao25Sum}, and the mobility-enabled adaptability of UAV networks \cite{Jiang21Resource,Lei24Trajectory,Xu25Collaborative} to enhance covertness.

Covert communication has also been studied in multi-user networks, particularly with respect to random user locations, interference, and spatial topology. He \textit{et al.} \cite{He18Covert} analyzed a homogeneous Poisson point process network in which numerous transmitters act as uncoordinated interferers. They showed that, in the interference-limited regime, the covert rate approaches a finite bound as the interference level increases, indicating that background interference does not provide unbounded covert-rate gains. Soltani \textit{et al.} \cite{Soltani18Covert} selected the friendly node closest to a warden to generate artificial noise, thereby increasing the warden’s detection uncertainty. Although these studies provide valuable insights into the statistical behavior of large-scale stochastic networks, they consider either uncoordinated ambient interference or individually selected friendly nodes, rather than coordinated and jointly optimized cooperation among multiple users.

More recent studies focus on covert communication assisted by multiple cooperative nodes, where legitimate users deliberately transmit interference signals to confuse a warden \cite{Zheng21Wireless,Lee23Multi,Lee24Channel,Yeom24Covert}. Zheng \textit{et al.} \cite{Zheng21Wireless} proposed a distributed cooperative jamming scheme in which friendly jammers are activated according to their instantaneous channel conditions toward the legitimate receiver, and jointly optimized the selection threshold and transmission rate to maximize the covert throughput. Although the individual helper-to-receiver and helper-to-warden distances are modeled, the activation rule depends only on the helper-to-receiver channel condition, and thus the warden-side geometry is not used to select the helpers that are effective in confusing the warden.
Building on the benefits of multi-user cooperation, Lee \textit{et al.}~\cite{Lee23Multi} further developed a joint power-control design for Alice and the cooperative users. In particular, they derived closed-form expressions that explicitly characterize the relationships among system parameters and proposed a low-complexity optimization algorithm. Later, \cite{Lee24Channel} incorporated channel correlation and proposed a scalable $Q$-learning-based algorithm. However, the frameworks in \cite{Lee23Multi,Lee24Channel} assume homogeneous large-scale channel conditions and therefore do not distinguish the cooperative users according to their individual user-to-warden large-scale fading conditions.
Yeom \textit{et al.} \cite{Yeom24Covert} further considered users located at unequal distances from Bob in the system model. However, for tractable analysis, the distinct user-to-Bob and user-to-warden channel variances were replaced by common average parameters, which prevents the resulting analytical design from retaining the individual spatial characteristics of the cooperative users. {\color{blue} These limitations are important because cooperative users that cause similar interference at Bob may provide substantially different interference at the warden. Consequently, a receiver-centric or averaged-geometry design may select ineffective users or inaccurately characterize the warden's aggregate-interference statistics, which can lead to an inaccurate estimate of the minimum detection error probability (DEP) and the minimum number of cooperative users. This may activate more cooperative users than necessary, increase the interference at Bob, and reduce the achievable covert rate. }

Motivated by these limitations, this work investigates covert communication in which a single covert user, Alice, is assisted by multiple spatially distributed non-covert users against a geometry-aware warden. 
Unlike the previous receiver-centric, homogeneous, and averaged-geometry frameworks, {\blue the proposed design retains the individual user-to-warden large-scale fading coefficients in both cooperative-user activation and warden's detection analysis, which enables the proposed framework to treat spatial heterogeneity as a design factor rather than merely as a deployment detail.} {\blue Accordingly, the main challenge lies not in simply substituting heterogeneous channel parameters into an existing model. Rather, it lies in developing a unified analytical framework in which each user’s warden-side geometry is incorporated into cooperative user activation, the resulting selected cooperative user set is used to characterize the warden’s aggregate-interference statistics and minimum DEP, and these analytical relationships are subsequently exploited to determine the minimum number of cooperative users and reformulate the original high-dimensional power-control problem.}
The major contributions are summarized as follows:
{\blue \begin{enumerate}
    \item \textbf{Spatial-heterogeneity-aware cooperative design:}
    Under the adopted analytical formulation for large-scale user deployments, we establish an on--off activation structure that jointly considers the interference caused at Bob and the interference effectiveness at the warden. Thus, the proposed design directly incorporates the individual warden-side geometry into cooperative-user selection. The same activation structure can also be implemented using estimated user-to-Bob channels, enabling the impact of channel estimation errors to be examined.
    \item \textbf{Detection and cooperation analysis under spatial heterogeneity:}
    We characterize the warden’s aggregate interference by retaining the individual large-scale fading conditions of the activated users and reflecting how the activation rule changes the composition of the cooperative user set. Based on this characterization, we derive a closed-form approximation of the DEP, the detection threshold that minimizes the DEP, and the minimum number of cooperative users required to satisfy the covert constraint. We further provide a compact approximation and recover the homogeneous deployment as a special case.
    \item \textbf{Low-complexity optimization algorithm:}
    Based on the derived relationships among Alice’s transmit power, the minimum number of cooperative users, and the activation threshold, we reduce the original high-dimensional problem to a one-dimensional piecewise search. The proposed algorithm evaluates only the analytically relevant candidates, with computational complexity dominated by a single sorting operation, and avoids the resolution dependence of conventional grid-based search.
    \item \textbf{Design insights under spatial heterogeneity:}
    The analysis shows that covertness depends on the balance between Alice’s received-energy shift and the uncertainty of the aggregate cooperative interference. Greater warden-side spatial dispersion reduces the minimum number of cooperative users, whereas a larger average user-to-warden distance increases it. Homogeneous or averaged-geometry analyses may therefore overestimate the cooperation requirement and cause unnecessary interference at Bob. The finite-sample DEP approaches the large-sample benchmark as the number of observations increases, while imperfect user-to-Bob channel estimation mainly reduces the achievable rate through cooperative user activation errors.
\end{enumerate}
}

The remainder of this paper is organized as follows. Section~\ref{Sec:System} presents the system model and problem formulation. Sections~\ref{Sec:On-OffScheme}--\ref{Sec:Optimization} establish the on--off activation structure, analyze the warden's detection performance, and develop the piecewise search algorithm, respectively. Section~\ref{Sec:Simulation} provides numerical validation, and Section~\ref{Sec:conclusion} concludes the paper.

\section{System Model and Problem Formulation} \label{Sec:System}

\subsection{System Model} \label{Subsec:SystemModel}

We consider a covert communication network where a single covert user, Alice, is assisted by multiple non-covert cooperative users. The network consists of Alice, \(M\) non-covert users, a base station Bob, and a warden Willie, as illustrated in Fig.~\ref{fig:System Model}. Alice transmits the covert signal over the covert carrier \(f_\mathrm a\). Each non-covert user \(U_m\) normally communicates with Bob over its own orthogonal carrier \(f_m\). A subset of non-covert users is selected as cooperative users to assist in hiding Alice’s signal from detection at Willie by transmitting interference over $f_\mathrm a$.

{\blue In this work, we consider a large-scale wireless network comprising a large number of static or slowly moving nodes, as commonly encountered in IoT deployments. Owing to the low mobility of the nodes, the channel coherence interval is assumed to be sufficiently long to accommodate the transmission of a long codeword. Furthermore, a large number of nodes are assumed to participate in the wireless communication process, i.e., $M \gg 1$. The channel coefficients are modeled as circularly symmetric complex Gaussian random variables, and hence, the corresponding channel magnitudes follow a Rayleigh distribution.} The channel from node $x$ to node $y$ at carrier frequency $z$ is denoted by 
\begin{equation}
 h_{x,y}^z \sim \mathcal{CN}(0,\lambda_{x,y}), \quad x,y \in \{ \mathrm a,\mathrm b,m,\mathrm w \}, z \in \{ \mathrm a,m \}   
\end{equation}
where $\lambda_{x,y}$ indicates the large-scale fading coefficient, and the indices $\mathrm a,\mathrm b,m$, and $\mathrm w$ represent Alice, Bob, $U_m$, and Willie, respectively. Now, for covert communication, we focus on the frequency band $f_\mathrm a$, and we omit the superscript. We assume a single-slope path-loss model, and thus, the large-scale fading coefficient is expressed as $\lambda_{x,y} = \beta_0 d_{x,y}^{-\alpha}$, where $\beta_0$ is the reference channel gain, $\alpha$ is the path-loss exponent, and $d_{x,y}$ denotes the Euclidean distances between nodes $x$ and $y$. The spatial heterogeneity considered in this work is captured by the user-dependent large-scale fading coefficients \(\{\lambda_{m,\mathrm b}\}_{m=1}^M\) and \(\{\lambda_{m,\mathrm w}\}_{m=1}^M\), which characterize the large-scale strengths of the cooperative links toward Bob and Willie, respectively.

The channel estimation process consists of two phases. First, Bob broadcasts pilot signals to all users, from which each node estimates the channel from Bob to itself. Second, each non-covert user transmits a secret orthogonal pilot signal over $f_\mathrm a$ \cite{Zheng21Wireless, Xu19Pilot}. The use of secret orthogonal pilot signals by each non-covert user enables Bob to estimate the channel between the non-covert user and Bob, whereas Willie cannot. The estimated channel at Bob is modeled as 
\begin{equation}\label{eq:est_error_model}
    \hat h_{m,\mathrm b} = \sqrt{1-\rho_e}\,h_{m,\mathrm b}+\sqrt{\rho_e}\,u_{m,\mathrm b},
\end{equation}
where \(u_{m,\mathrm b}\sim\mathcal{CN}(0,\lambda_{m,\mathrm b})\) is independent of \(h_{m,\mathrm b}\), and \(\rho_e\in[0,1]\) denotes the channel estimation error level.

In addition, $\mathcal{H}_0$ and $\mathcal{H}_1$ indicate the hypotheses that Alice does not transmit or transmits a covert message, respectively. 
The received signal at Bob $y[t]$ is given by
\begin{align} \label{eq:y_Bob}
    y[t]= 
  \begin{cases}
    \displaystyle\sum_{m = 1}^{M} \sqrt{P_m} h_{m,\mathrm b} x_m[t] + n_\mathrm b[t], & \mathcal{H}_0 \\
    \displaystyle\sqrt{P_\mathrm a}h_{\mathrm a,\mathrm b} x_\mathrm a[t] \\ ~~~~+ \sum_{m = 1}^{M} \sqrt{P_m}h_{m,\mathrm b}x_m[t] + n_\mathrm b[t], & \mathcal{H}_1
  \end{cases},
\end{align}
where both $x_\mathrm a[t]$ and $x_m[t]$ are assumed to be complex Gaussian random variables with zero-mean unit-variance. The parameters $P_\mathrm a$ and $P_m$ represent the transmit power of Alice and $U_m$, respectively, and $n_\mathrm b[t] \sim \mathcal{CN}(0, \sigma_\mathrm b^2)$ denotes AWGN at Bob. 
Similarly, the received signal at Willie over $f_\mathrm a$ is given by
\begin{equation} \label{eq:y_Willie}
    z[t]= 
  \begin{cases}
      \displaystyle\sum_{m = 1}^{M} \sqrt{P_m} h_{m,\mathrm w} x_m[t] + n_\mathrm w[t], & \mathcal{H}_0 \\
      \sqrt{P_\mathrm a}h_{\mathrm a, \mathrm w} x_\mathrm a[t] \\ ~~~~+ \sum_{m = 1}^{M} \sqrt{P_m} h_{m,\mathrm w} x_m[t] + n_\mathrm w[t], & \mathcal{H}_1
  \end{cases},
\end{equation}
where $n_\mathrm w[t] \sim \mathcal{CN}(0, \sigma_\mathrm w^2)$ stands for AWGN at Willie.

\begin{figure}[t]
\centering
\includegraphics[width=\widefigwidth]{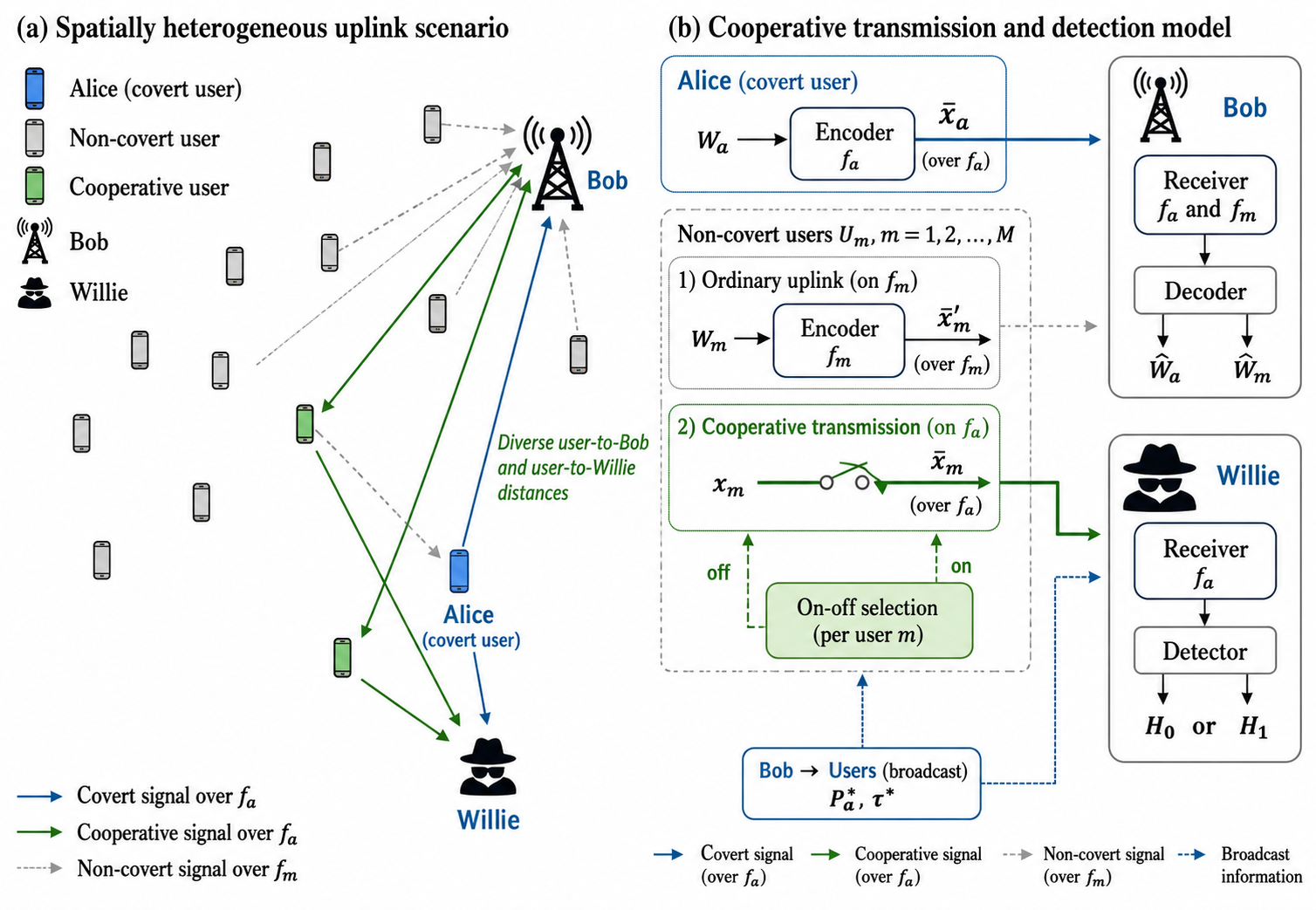}
\caption{Illustration of the proposed system model. (a) Spatially heterogeneous uplink
scenario. (b) Cooperative transmission and detection model.}
\label{fig:System Model} 
\end{figure}

\subsection{Detection at Willie}
 
Willie infers the presence of covert data transmission from Alice based on his observation $\mathbf{z}=[z[1],z[2],\ldots,z[N]]^T$. As introduced in Section \ref{Subsec:SystemModel}, Willie is confronted with a binary hypothesis testing problem consisting of two hypotheses $\mathcal{H}_0$ and $\mathcal{H}_1$. {\blue Let $\mathcal{D}_1$ and $\mathcal{D}_0$ denote Willie's decisions in favor of $\mathcal{H}_1$ and $\mathcal{H}_0$, respectively. For a detection threshold $\gamma$, the false-alarm and miss-detection probabilities are defined as
\[
P_{\rm FA}(\gamma)=\Pr(\mathcal D_1|\mathcal H_0),
\qquad
P_{\rm MD}(\gamma)=\Pr(\mathcal D_0|\mathcal H_1).
\]
} 
Under the conventional equal prior assumption \cite{Bash16Covert, Sobers17Covert}, the DEP can be defined as
\begin{equation}
\zeta \triangleq P_{\mathrm{FA}}(\gamma) + P_{\mathrm{MD}}(\gamma).
\end{equation}
Since Willie chooses the threshold that minimizes the DEP, the covert constraint is imposed on the minimum DEP, i.e., \(\zeta_{\min}\ge 1-\epsilon\), where \(\epsilon>0\) is arbitrarily small~\cite{Hu19Covert,Lee23Multi,Lee24Channel}.

According to the Neyman–Pearson Lemma \cite{Neyman33OntheProblem}, the likelihood ratio test (LRT) is optimal for the considered binary hypothesis test and is expressed as 
\begin{equation}
    \frac{\textstyle\prod_{t=1}^{N} f\bigl(z[t]; \mathcal{H}_1\bigr)}
           {\textstyle\prod_{t=1}^{N} f\bigl(z[t]; \mathcal{H}_0\bigr)}
    \mathop{\gtrless}_{\mathcal{D}_0}^{\mathcal{D}_1} \gamma,
\end{equation}
where $f(\cdot)$ represents the probability density function (PDF) and $N$ denotes the codeword length (number of channel uses). 
For a general cooperative power profile \(\mathbf P=(P_1,\ldots,P_M)\), define the instantaneous received powers at Willie as
\(
\rho_0(\mathbf P) = \sum_{m=1}^{M}P_m|h_{m,\mathrm w}|^2+\sigma_{\mathrm w}^2, 
\) and
\( \rho_1(\mathbf P) = \rho_0(\mathbf P)+P_{\mathrm a}|h_{\mathrm a,\mathrm w}|^2. 
\)
Under the adopted Gaussian signaling model, conditioned on a cooperative power profile \(\mathbf P\), Willie's received samples satisfy $ z[t]\mid \mathcal H_i,\mathbf P \sim \mathcal{CN}\bigl(0,\rho_i(\mathbf P)\bigr)$, for $i\in\{0,1\}$.
{\blue Hence, the conditional likelihood can be written as $p(\mathbf z\mid \mathcal H_i,\mathbf P)
= \bigl(\pi\rho_i(\mathbf P)\bigr)^{-N} \exp\left(-N T_{\mathrm w}/\rho_i(\mathbf P)\right)$, where \(T_{\mathrm w}=\sum_{t=1}^{N}|z[t]|^2 / N\). Therefore, the likelihood depends on \(\mathbf z\) only through \(T_{\mathrm w}\), and \(T_{\mathrm w}\) is a sufficient statistic for Willie's LRT.
Moreover, since $\rho_1(\mathbf P) > \rho_0(\mathbf P)$, the conditional likelihood ratio is increasing monotonically in $T_\mathrm w$.
Consequently, the LRT can be equivalently reduced to a test of the received energy, i.e., \( T_\mathrm w \mathop{\gtrless}_{\mathcal{D}_0}^{\mathcal{D}_1} \gamma. \)  } 
For finite \(N\), conditioned on \(\rho_i(\mathbf P)\) under \(\mathcal H_i\), \(i\in\{0,1\}\), the received-energy statistic satisfies
\begin{equation} \label{eq:finiteN_statistic}
T_{\rm w}\mid \rho_i(\mathbf P) = \rho_i(\mathbf P)\frac{\chi^2_{2N}}{2N}, \qquad i\in\{0,1\}, 
\end{equation} 
where \(\chi^2_{2N}\) denotes a chi-square random variable with \(2N\) degrees of freedom. This expression explicitly captures the finite-sample fluctuation of Willie's energy statistic within a fading block.
Using the test of received energy, the finite-\(N\) DEP for a given detection threshold \(\gamma\) is written as \begin{equation}\label{eq:DEP} 
\zeta_N(\gamma) = \Pr(T_{\rm w}>\gamma|\mathcal H_0) + \Pr(T_{\rm w}\le\gamma|\mathcal H_1). 
\end{equation} 
Willie is assumed to choose the threshold that minimizes the DEP, i.e., \begin{equation} 
\zeta_{N,\min} = \min_{\gamma}\zeta_N(\gamma). 
\end{equation}

{\blue The closed-form analysis in the following sections uses the large-\(N\) benchmark, i.e., $\zeta_{\min}$, where \(\chi^2_{2N}/(2N)\to 1\) almost surely \cite[Sec.~5.5]{Casella02StatInference}. This benchmark gives Willie an asymptotically accurate estimate of the instantaneous received power $\rho_i$ through $T_{\rm w}$ and thus is conservative from the viewpoint of the legitimate users, i.e., Alice and Bob. The finite-\(N\) DEP is separately evaluated in the numerical results to quantify practical finite-sample detection effects.} 
Even for a geometry-aware Willie, the aggregate interference remains random due to small-scale fading and power profile uncertainty, with its statistics shaped by the users' spatial heterogeneity.

\subsection{Optimization Problem}

{\blue This work considers an instantaneous transmission-rate maximization problem under a covert constraint. 
The communication parameters are adapted to the channel realization available at Bob, and hence the rate metric used in this work should be interpreted as an instantaneous achievable transmission rate conditioned on the realized channels.} For a given channel realization and transmit-power profile, the instantaneous achievable transmission rate of Alice is defined as
\begin{equation}\label{eq:rate_def}
R = \log_2\left(1 + 
\frac{P_\mathrm a |h_\mathrm{a,b}|^2}
{\sum_{m=1}^M P_m |h_{m,\mathrm b}|^2 + \sigma_\mathrm b^2}
\right),
\end{equation}
where the denominator includes the aggregate interference caused at Bob by the cooperative transmissions of the non-covert users.
Then, the transmit power optimization problem can be formulated as
\begin{subequations}\label{eq:original_opt}
\begin{align}\label{eq:rate_opt}
&(P_\mathrm a^*, \mathbf{P}^*) = \argmax_{P_\mathrm a, \mathbf{P}} ~ R \\
& ~~~~\text{s.t.} ~~ \zeta_{\min} \geq 1- \epsilon \label{eq:CovertConst}\\
& ~~~~~~~~~ 0 \leq P_\mathrm a \leq P_{\max} \label{eq:PowConst_Alice}\\
& ~~~~~~~~~ 0 \leq P_m \leq P_{\max}, \quad \forall m, \label{eq:PowConst_Noncovert}
\end{align}
\end{subequations}
where $P_{\max}$ is the maximum transmit power and $\epsilon$ is the tolerable covertness error. The covert constraint in \eqref{eq:CovertConst} is imposed through Willie’s minimum DEP, where Willie chooses the detection threshold that minimizes his detection error probability, and the transmit power constraints are given in \eqref{eq:PowConst_Alice} and \eqref{eq:PowConst_Noncovert}, respectively.

{\blue The problem in \eqref{eq:original_opt} is challenging for several reasons. 
First, it jointly optimizes Alice's transmit power and the $M$-dimensional non-covert user power profile, so the number of optimization variables grows linearly with the network size. For the considered networks with hundreds to thousands of users, a direct search over this $(M+1)$-dimensional space is computationally prohibitive. 
Second, the objective function has a fractional SINR structure in which the two sets of variables act in opposite directions. $P_{\mathrm{a}}$ enters the numerator of Bob's SINR, while $\mathbf{P}$ enters its denominator, and the covert constraint couples them in the reverse sense. That is, increasing $P_{\mathrm{a}}$ improves Alice's received signal
power at Bob but makes the covert transmission more distinguishable at Willie, whereas increasing $\mathbf{P}$ improves covertness by increasing the uncertainty of the aggregate interference at Willie but degrades Bob's received SINR. 
Third, the covert constraint is not a simple convex power or SINR constraint. It is a statistical constraint determined after Willie optimizes his detection threshold. Finally, under spatial heterogeneity, the aggregate cooperative interference at Willie depends on the selected users' individual large-scale fading coefficients $\{\lambda_{m,\mathrm{w}}\}_{m=1}^{M}$.} 
These challenges motivate the structural analysis of power profile for the cooperative user and Willie’s detection performance in the following sections.

\section{On--Off Cooperative User Activation} \label{Sec:On-OffScheme}

In this section, we analyze the cooperative user activation structure for a given Alice transmit power $P_\mathrm a$. 
The purpose is to reduce the $M$-dimensional non-covert-user power profile $\mathbf P=(P_1,\ldots,P_M)$ in \eqref{eq:original_opt} to a more tractable and simpler form. 
The structural result derived in this section is established under the adopted large-$M$ analytical formulation for the fixed $P_\mathrm a$ problem. 
Then, the original optimization problem is reformulated in terms of Alice’s transmit power $P_\mathrm a$ and the activation threshold $\tau$.
For notational convenience, define the activation metric of the $m$-th non-covert user as
\begin{equation}\label{eq:activation_metric}
    r_m \triangleq \frac{|h_{m,\mathrm b}|^2}{\lambda_{m,\mathrm w}},
    \quad m=1,\ldots,M.
\end{equation}
{\blue A smaller value of $r_m$ indicates that the user causes relatively weak interference at Bob while providing a relatively strong interference contribution at Willie.}

\begin{prop}\label{prop:power_profile}
{\blue Under the adopted large-$M$ analytical formulation, for a given $P_\mathrm a$, the cooperative-user power profile that minimizes Bob-side interference while satisfying the covert constraint has the following on--off activation structure}, i.e., 
\begin{equation} \label{eq:onoff_prop}
       P_m^*= 
  \begin{cases}
      P_{\max}, & r_m \leq \tau \\
      0 , & \text{otherwise}
  \end{cases},
\end{equation}
{\blue where $\tau$ denotes the activation threshold.}
\end{prop}
\begin{proof}
The proof is provided in Appendix \ref{Appendix:power_profile}.
\end{proof}

Based on the on--off scheme presented in Proposition~\ref{prop:power_profile}, the activation threshold $\tau$ determines the set of cooperative users as
\begin{equation}\label{eq:coop_set}
    \mathcal S=\{m:r_m\leq \tau\},
\end{equation}
and the number of cooperative users is given by $K=|\mathcal S|$.
Increasing $\tau$ activates more non-covert users and increases the aggregate interference observed at Willie, which improves the concealment of Alice's covert signal. However, it also increases the interference at Bob and degrades the achievable rate. Therefore, for a given $P_\mathrm a$, the threshold should be selected as the smallest value that satisfies the covert constraint.
After the optimization, Bob broadcasts the control parameters, including $P_\mathrm a^\star$ and $\tau^\star = \tau(P_\mathrm a^\star)$, to the users. 
{\blue The on--off cooperation rule and the broadcast control parameters, including $P_{\mathrm{a}}^{\star}$ and $\tau^{\star}$, are not treated as secret information and can be overheard by Willie without any active participation in the network. 
Since $\lambda_{m,\mathrm{w}}$ is a long-term geometry-dependent quantity, Willie can infer it by passively overhearing the regular uplink transmissions of the non-covert users, where the long-term average received power averages out the small-scale fading. Meanwhile, the instantaneous user-to-Bob channel gains are estimated through the secret orthogonal pilots described in Section~II-A and are thus not observable by Willie. Hence, the exact value of $K$ is granted to Willie as a conservative worst-case assumption.} Nevertheless, the set of cooperative users and the aggregate interference remain random because the activation rule depends on the non-covert user-to-Bob channel, $h_{m,\mathrm b}$. Furthermore, the proposed on--off rule is different from those in existing studies~\cite{Lee23Multi,Lee24Channel,Yeom24Covert}. The existing schemes mainly determine the cooperative users based on the user-to-Bob channel gains. In contrast, the proposed rule uses $r_m$, which jointly accounts for the interference caused at Bob and the interference observed by Willie.
Based on Proposition~\ref{prop:power_profile}, the original optimization problem in \eqref{eq:original_opt} can be reformulated as
\begin{subequations} \label{eq:reformul_opt}
\begin{align}
  \left(P_\mathrm a^\star, \tau^\star \right)
  &= \argmax_{P_\mathrm a,\tau} \quad R \\
  \text{s.t.} \quad
  & \zeta_{\min} \geq 1-\epsilon, \\
  & 0 \leq P_\mathrm a \leq P_{\max}.
\end{align}
\end{subequations}
Note that the reformulated problem in \eqref{eq:reformul_opt} has only two optimization variables, whereas the original problem in \eqref{eq:original_opt} involves $M+1$ optimization variables. Consequently, owing to Proposition~\ref{prop:power_profile}, a much simpler optimization problem can be obtained.

{\blue
\begin{remark}
When Bob has imperfect estimates of non-covert user-to-Bob channels, the same on--off structure can be implemented using the estimated activation metric, $\hat r_m=|\hat h_{m,\mathrm b}|^2/\lambda_{m,\mathrm w}$. Then, Bob applies the on--off rule by replacing $r_m$ with $\hat r_m$. This preserves the on--off cooperative user activation structure. However, the channel estimation error can make $\hat{\mathcal S}=\{m:\hat r_m\leq\tau\}$ different from $\mathcal S =\{m:r_m\leq\tau\}$. If some users satisfy $\hat r_m\leq\tau$ but $r_m>\tau$, they are activated although they are unfavorable under the true channel, thereby increasing the actual interference at Bob and reducing the achievable rate.
\end{remark}
}
 
The next section analyzes Willie's detection performance by accounting for the interference statistics induced by the activation rule.

\section{Analysis of Detection at Willie} \label{Sec:DEP_Willie}

In this section, we derive a closed-form expression of the minimum DEP and analyze its characteristics. In particular, we determine the optimal detection threshold $\gamma^*$ that minimizes the DEP, resulting in the closed-form expression of $\zeta_{\min}$. Furthermore, based on the closed-form expression of $\zeta_{\min}$, we derive the minimum number of cooperative users and the corresponding activation threshold required to satisfy the covert constraint, i.e.,
$\zeta_{\min} \ge 1-\epsilon$, both in closed-form expressions.

According to the on--off scheme in Proposition \ref{prop:power_profile}, the cooperative users transmit interference signals with $P_m = P_{\rm{max}}$, while the other users do not transmit any interference signals, i.e., $P_m = 0$. Then, under the on--off scheme in Proposition \ref{prop:power_profile}, the test statistic, $T_\mathrm w$ is expressed as 
\begin{align}\label{eq:test_statistic2}
   T_\mathrm w = 
  \begin{cases}
    \sum_{i=1}^{K} P_{\max} |h_{m_i,\mathrm w}|^2 + \sigma_\mathrm w^2, & \mathcal{H}_0, \\[4pt]
    \sum_{i=1}^{K} P_{\max} |h_{m_i,\mathrm w}|^2 + P_\mathrm a |h_{\mathrm a,\mathrm w}|^2 + \sigma_\mathrm w^2, & \mathcal{H}_1,
  \end{cases}
\end{align}
where $m_i \in \{1,2,\ldots, M\}$ denotes the index of the $i$-th cooperative user sorted in ascending order of the ratio between the channel gain and the large-scale fading coefficient, i.e., $r_{m_i} \ge r_{m_j}$ for $i \ge j$. It is noteworthy that the selected indices $\{m_i\}_{i=1}^{K}$ are determined by this ordering rule and are not generally equivalent to a uniformly random subset of size $K$.
Conditioned on the selected indices $\{m_i\}_{i=1}^{K}$, since each channel gain $|h_{m_i,\mathrm w}|^2$ follows a Gamma distribution, and the sum of independent random variables is represented by the convolution of their distributions, the shifted test statistic $T_\mathrm w-\sigma_\mathrm w^2$ is expressed as
\begin{align}\label{eq:test_statistic_RV}
  T_\mathrm w - \sigma_\mathrm w^2 \sim
   \begin{cases}
      \Conv_{i=1}^K \Gamma(1, \lambda_{m_i,\mathrm w} P_{\max}), & \mathcal{H}_0, \\[4pt]
      \Conv_{i=1}^K \Gamma(1, \lambda_{m_i,\mathrm w} P_{\max}) \ast \Gamma(1, \lambda_\mathrm {a,w}P_\mathrm a), & \mathcal{H}_1,
   \end{cases}
\end{align}
where $\Gamma(a,b)$ represents a Gamma distribution with shape parameter $a$ and scale parameter $b$, $\ast$ denotes the convolution operator between PDFs of two random variables, and $\Conv_{i=1}^K$ indicates the successive convolution of the PDFs of $K$ random variables, i.e., $\Conv_{k=1}^K A_k = A_1 \ast A_2 \ast \cdots \ast A_K$. Based on the probabilistic distribution of $T_\mathrm w - \sigma_\mathrm w^2$ in \eqref{eq:test_statistic_RV}, a closed-form expression for the DEP in \eqref{eq:DEP} is derived in Lemma~\ref{lemma:DEP}.

\begin{lemma}\label{lemma:DEP}
With the on--off scheme, the DEP for a given number of cooperative users, \(K\), is approximated as
\begin{align}\label{eq:DEP_general}
  \zeta(\gamma)
  &\approx
  1-
  \exp\!\left(
    -\frac{2\Delta(\hat{\gamma}-\Xi_K)-\Sigma_K}{2\Delta^2}
  \right) \nonumber \\
  &\hspace{5em}
 \times Q\!\left(
    -\frac{\Delta(\hat{\gamma}-\Xi_K)-\Sigma_K}
    {\sqrt{\Sigma_K}\Delta}
  \right),
\end{align}
where $Q(\cdot)$ denotes the standard Gaussian $Q$-function, 
\(\hat{\gamma}\triangleq\gamma-\sigma_{\mathrm w}^2\),  
\( \Xi_K \triangleq  P_{\max} \sum_{m=1}^{M}   \pi_m^{(K)}\lambda_{m,\mathrm w}\), 
\(\Delta\triangleq P_{\mathrm a}\lambda_{\mathrm a,\mathrm w}\), and
\begin{align}
   \Sigma_K
    &\triangleq
    P_{\max}^{2}
    \Bigg[
    \sum_{m=1}^{M}
    \pi_m^{(K)}\lambda_{m,\mathrm w}^{2}
    +
    \sum_{m=1}^{M}
    \pi_m^{(K)}
    \bigl(1-\pi_m^{(K)}\bigr)
    \lambda_{m,\mathrm w}^{2}
    \nonumber\\
    &\hspace{1em}
    +
    2\sum_{1\le m<n\le M}
    \left(
    \pi_{m,n}^{(K)}
    -
    \pi_m^{(K)}\pi_n^{(K)}
    \right)
    \lambda_{m,\mathrm w}\lambda_{n,\mathrm w}
    \Bigg].
    \label{eq:Sigma_K_selection}
\end{align}
{\blue Here, \(\pi_m^{(K)}\) and \(\pi_{m,n}^{(K)}\) are the first- and second-order inclusion probabilities, respectively, defined in \eqref{eq:first_order_inclusion_probability} and \eqref{eq:second_order_inclusion_probability}.}
\end{lemma}

\begin{proof}
See Appendix \ref{Appendix:DEP}
\end{proof}

The analytical form of the DEP in \eqref{eq:DEP_general} enables one to obtain several design parameters. In particular, the detection threshold \(\gamma^\ast\), which minimizes the DEP, is derived in closed form and presented in Lemma~\ref{lemma:gamma}.

\begin{lemma}\label{lemma:gamma}
The detection threshold that minimizes \(\zeta(\gamma)\) is given by
\begin{equation}\label{eq:Opt_gamma}
  \gamma^\ast = \Xi_K+\sigma_{\mathrm w}^2.
\end{equation}
\end{lemma}

\begin{proof}
Starting from \eqref{eq:DEP_general}, the DEP is approximated as
\begin{equation}
  \zeta (\gamma) \approx
  1 - c_q \exp\left( \underbrace{  -\frac{1}{2\Sigma_K} \left( \hat\gamma^2-2\Xi_K \hat\gamma + \Xi_K^2 \right)}_{(a)}  \right), \label{eq:ApproxDEP}
\end{equation}
where the approximation follows from the exponential approximation of the Gaussian \(Q\)-function, i.e., \(Q(x)\approx c_q e^{-x^2/2}\)~\cite{Lee23Multi}, \cite[Eqs.(4)-(5)]{Chiani03New}. It is observed that \eqref{eq:ApproxDEP} is an exponential function of a concave quadratic term in \(\hat{\gamma}\). Since the term (a) attains its maximum value at \(\hat{\gamma}=\Xi_K\), the DEP is minimized at \(\hat{\gamma}_K^\ast=\Xi_K\). From \(\hat{\gamma}=\gamma-\sigma_{\mathrm w}^2\), the detection threshold is obtained as \eqref{eq:Opt_gamma}.
\end{proof}

Using the detection threshold in \eqref{eq:Opt_gamma}, the minimum DEP is obtained as presented in Lemma \ref{lemma:Min_DEP}.
\begin{lemma}\label{lemma:Min_DEP}
The minimum DEP at Willie is given by
\begin{equation}\label{eq:Opt_DEP}
  \zeta_{\min}
  = 1 - \left( \sqrt{\pi}
  \left(
  \sqrt{\frac{\Sigma_K}{2\Delta^2}}
  + 
  \sqrt{\frac{\Sigma_K}{2\Delta^2} + \frac{4}{\pi}}
  \right) \right)^{-1}.
\end{equation}
\end{lemma}

\begin{proof}
By substituting the optimal detection threshold $\gamma^\ast$ in \eqref{eq:Opt_gamma} into the DEP in \eqref{eq:DEP_general}, we obtain
\begin{align} 
    \zeta_{\min} &= 1-\exp\!\left(\frac{\Sigma_K}{2\Delta^2}\right)Q\!\left(\frac{\sqrt{\Sigma_K}}{\Delta}\right) \nonumber\\
    &  = 1 - \frac{1}{2}\exp\!\left(\frac{\Sigma_K}{2\Delta^2}\right)\mathrm{erfc}\!\left(\frac{\sqrt{\Sigma_K}}{\sqrt{2}\,\Delta}\right) \nonumber \\
    & \mathop{\approx}_{{\mathrm{(a)}}} 1 - \frac{1}{2}\exp \left(\frac{\Sigma_K}{2\Delta^2} \right)\frac{2\exp \left( -\frac{\Sigma_K}{2\Delta^2} \right) }{\sqrt{\pi}\left(\sqrt{\frac{\Sigma_K}{2\Delta^2}} + \sqrt{\frac{\Sigma_K}{2\Delta^2} + \frac{4}{\pi}} \right)} \nonumber \\
    & = 1-\left( \sqrt{\pi}\left(\sqrt{\frac{\Sigma_K}{2\Delta^2}} + \sqrt{\frac{\Sigma_K}{2\Delta^2} + \frac{4}{\pi}} \right) \right) ^{-1} ,
\end{align}
where $\mathrm{erfc}(\cdot)$ indicates the complementary error function. 
{\blue The approximation in (a) follows from the fact that 
$\operatorname{erfc}(x) \approx 2/\sqrt{\pi} \cdot \operatorname{exp}(-x^2)/(x + \sqrt{x^2 + 4/\pi})$ is highly accurate for $x \gtrsim 1$~\cite[Eq.~(7.1.13)]{Abramowitz64Handbook}. In our setting, the argument of the complementary error function is $\sqrt{\Sigma_K }/( \sqrt{2}\Delta)$, and since the aggregate interference uncertainty $\Sigma_K$ from the cooperative users is typically much larger than that of Alice's power $\Delta^2$, this approximation is well justified. }
\end{proof}

Based on Lemma~\ref{lemma:Min_DEP}, the number of cooperative users should be chosen as the smallest integer that satisfies the covert constraint \(\zeta_{\min}\ge 1-\epsilon\). Since \(\zeta_{\min}\) is an increasing function of \(\Sigma_K\), this condition can be equivalently expressed in terms of the required aggregate interference uncertainty. The resulting minimum number of cooperative users is characterized in Theorem~\ref{Theorem:K}.

\begin{theorem}\label{Theorem:K}
Under the DEP approximation in Lemma~\ref{lemma:DEP}, the minimum number of cooperative users required to satisfy the covert constraint is given by
\begin{equation}\label{eq:Kmin_selection}
    K_{\min}
    =
    \min
    \left\{
    K\in\{0,1,\ldots,M\}:
    \Sigma_K \ge c_\epsilon\Delta^2
    \right\},
\end{equation}
where \( c_\epsilon \triangleq \tfrac{1/\epsilon^2 -8+16\epsilon^2}{2\pi}\).
\end{theorem}
\begin{proof}

By Lemma~\ref{lemma:Min_DEP}, the covert constraint \(\zeta_{\min}\ge 1-\epsilon\) is equivalent to
\begin{equation}
\left[
  \sqrt{\pi}
  \left(
  \sqrt{\frac{\Sigma_K}{2\Delta^2}}
  +
  \sqrt{\frac{\Sigma_K}{2\Delta^2}+\frac{4}{\pi}}
  \right)
\right]^{-1}
\le \epsilon.
\end{equation}
Solving this inequality with respect to \(\Sigma_K\) yields \(\Sigma_K \ge c_\epsilon\Delta^2\).
Hence, the minimum number of cooperative users is the smallest integer \(K\) satisfying this condition, which gives \eqref{eq:Kmin_selection}.
\end{proof}

To obtain a compact closed-form rule, we approximate the selected users as a uniformly chosen \(K\) subset without replacement, for which the finite population correction applies~\cite[Sec.~2.6]{Cochran77Sampling},\cite[Sec.~2.3]{Lohr19Sampling}. Then, \begin{equation}\label{eq:uniform_inclusion_prob} 
\pi_m^{(K)} \approx \frac{K}{M}, \qquad \pi_{m,n}^{(K)} \approx \frac{K(K-1)}{M(M-1)}. 
\end{equation} 
Substituting \eqref{eq:uniform_inclusion_prob} into \(\Sigma_K\) gives
\begin{equation}\label{eq:Sigma_K_uni}
\Sigma_K^{\rm uni}
=
P_{\max}^{2}
\left[
KE+
K\frac{M-K}{M-1}V
\right],
\end{equation}
where $ E\triangleq \tfrac{1}{M}\textstyle\sum_{m=1}^{M}\lambda_{m,\mathrm w}^{2}$ and
\begin{equation*}
V\triangleq
\tfrac{1}{M}\textstyle\sum_{m=1}^{M}
\left(
\lambda_{m,\mathrm w}
-
\tfrac{1}{M}\textstyle \sum_{n=1}^{M}\lambda_{n,\mathrm w}
\right)^2 .
\end{equation*}
{\blue Let \(K_{\min}^{\rm uni}\) be the minimum \(K\) satisfying
\(\Sigma_K^{\rm uni}\ge c_\epsilon\Delta^2\). This closed-form rule is
conservative for a given topology whenever
\begin{equation}\label{eq:sufficient_condition_uniform}
\Sigma_{K_{\min}^{\rm uni}}
\ge
\Sigma_{K_{\min}^{\rm uni}}^{\rm uni}.
\end{equation}
Indeed, under \eqref{eq:sufficient_condition_uniform},
\begin{equation}\label{eq:sufficient_condition_chain}
\Sigma_{K_{\min}^{\rm uni}}
\ge
\Sigma_{K_{\min}^{\rm uni}}^{\rm uni}
\ge
c_\epsilon\Delta^2,
\end{equation}
and thus \(K_{\min}^{\rm uni}\) also satisfies the general covert constraint in Theorem~\ref{Theorem:K}. Hence, the following corollary should be understood as a sufficient closed-form design rule under the uniform finite-population approximation.}

\begin{corollary}\label{Corollary:K_uniform}
Under the uniform finite-population approximation in \eqref{eq:Sigma_K_uni}, the minimum number of cooperative users can be approximated in closed form as
\begin{equation}\label{eq:K1_exact_approx}
K_{\min}^{\rm uni}
=
\left\lceil
\tfrac{M(E+V)}{2V}
\left(
1-
\sqrt{
1-\tfrac{4V}{CM(E+V)^2}
}
\right)
\right\rceil,
\end{equation}
where \(C \triangleq P_{\max}^{2}/(P_{\mathrm a}^{2}\lambda_{\mathrm a,\mathrm w}^{2}c_\epsilon)\).
\end{corollary}

\begin{proof}
From Theorem~\ref{Theorem:K}, the covert constraint is satisfied if \( \Sigma_K \ge c_\epsilon\Delta^2\).

By substituting the uniform finite-population approximation in \eqref{eq:Sigma_K_uni} and using \(\Delta=P_{\mathrm a}\lambda_{\mathrm a,\mathrm w}\) together with the definition of \(C\), this condition is rewritten as
\begin{equation}
K\left(
E+\frac{M-K}{M-1}V
\right)
\ge
\frac{1}{C}.
\end{equation}
Solving the corresponding quadratic inequality with respect to \(K\) and taking the smallest integer satisfying the inequality yields \eqref{eq:K1_exact_approx}.
\end{proof}
The expression in Corollary~\ref{Corollary:K_uniform} provides the same physical insight as the original scaling law. Since we assume \(M\gg 1\), applying \(\sqrt{1-x}\approx1-x/2\) to \eqref{eq:K1_exact_approx} yields 
\begin{equation}\label{eq:Kmin_uni_scaling} 
K_{\min}^{\rm uni} \approx \left\lceil \frac{1}{C(E+V)} \right\rceil = \left\lceil \frac{ c_\epsilon P_{\mathrm a}^{2}\lambda_{\mathrm a,\mathrm w}^{2} } { P_{\max}^{2}(E+V) } \right\rceil . 
\end{equation} 
{\blue This scaling shows that the minimum number of cooperative users decreases with the Willie-side second-order spatial statistic \(E+V\), where \(E\) captures the empirical second moment of \(\lambda_{m,\mathrm w}\) and \(V\) captures spatial heterogeneity. Hence, stronger spatial heterogeneity increases the aggregate cooperative-interference uncertainty observed by Willie, making Alice's additional energy easier to hide.} 
Conversely, the minimum number of cooperative users increases quadratically with \(P_{\mathrm a}\lambda_{\mathrm a,\mathrm w}\). In addition, since \(c_\epsilon=\mathcal O(\epsilon^{-2})\) as \(\epsilon\to0\), the minimum number of cooperative users also increases on the order of
\(\epsilon^{-2}\) as the covert constraint becomes tighter. These trends are numerically confirmed in Section~\ref{Sec:Simulation}.

When \(\lambda_{m,\mathrm w}\approx\bar{\lambda}_{\mathrm w}\) for all \(m\), the spatial heterogeneity vanishes, and the selected cooperative set no longer introduces additional uncertainty through large-scale heterogeneous fading from non-covert user-to-Willie. Therefore, \(\Sigma_K\approx P_{\max}^{2}K\bar{\lambda}_{\mathrm w}^{2}\), which gives the following homogeneous special case from Theorem~\ref{Theorem:K}.

\begin{corollary}\label{cor:homogeneous}
Under homogeneous user deployment, the minimum number of
cooperative users is approximated as
\begin{equation}\label{eq:Kmin_simplified}
K_{\min}^{\rm hom}
\approx
\left\lceil
\frac{P_{\mathrm a}^{2}c_\epsilon}{P_{\max}^{2}}
\left(
\frac{\lambda_{\mathrm a,\mathrm w}}
{\bar{\lambda}_{\mathrm w}}
\right)^2
\right\rceil .
\end{equation}
\end{corollary}

\begin{proof}

Under homogeneous deployment, \(\lambda_{m,\mathrm w}\approx\bar{\lambda}_{\mathrm w}\)
for all \(m\), and hence \(\Sigma_K\approx P_{\max}^{2}K\bar{\lambda}_{\mathrm w}^{2}\).
Substituting this expression into the condition
\(\Sigma_K\ge c_\epsilon\Delta^2\) in Theorem~\ref{Theorem:K}, with
\(\Delta=P_{\mathrm a}\lambda_{\mathrm a,\mathrm w}\), gives
\begin{equation}
    K
    \ge
    \frac{
    c_\epsilon P_{\mathrm a}^{2}\lambda_{\mathrm a,\mathrm w}^{2}
    }
    {
    P_{\max}^{2}\bar{\lambda}_{\mathrm w}^{2}
    }.
\end{equation}
Applying the integer constraint yields \eqref{eq:Kmin_simplified}.
\end{proof}

Equation~\eqref{eq:Kmin_simplified} shows that, in the homogeneous case, the minimum number of cooperative users is determined by the relative Willie-side channel between Alice and the cooperative users. In addition, when $\lambda_{a,w} = \bar{\lambda}_w$, this result directly reduces to the result shown in previous works \cite{Lee23Multi, Yeom24Covert}, demonstrating that the proposed analysis encompasses these prior studies as special cases and therefore provides a more general characterization of the system.

By utilizing the minimum number of cooperative users obtained from Theorem~\ref{Theorem:K}, or its closed-form approximation in Corollary~\ref{Corollary:K_uniform}, the corresponding activation threshold of the on--off scheme in Proposition~\ref{prop:power_profile} can be readily determined, as presented in Theorem~\ref{theorem:tau}.
\begin{theorem}\label{theorem:tau}
The activation threshold is determined by 
\begin{equation}
   \tau
    =
    \frac{|h_{m_{K_{\min}},\mathrm b}|^{2}}
    {\lambda_{m_{K_{\min}},\mathrm w}} .
\end{equation} 
\end{theorem}

\begin{proof}
The on--off scheme activates user \(m\) if \(|h_{m,\mathrm b}|^2/\lambda_{m,\mathrm w}\le\tau\). Therefore, setting \(\tau\) to the activation metric of user \(m_{K_{\min}}\) makes the \(K_{\min}\)-th smallest metric the activation boundary, ensuring that exactly \(K_{\min}\) cooperative users participate.
\end{proof}

\section{Optimization of System Parameters}\label{Sec:Optimization}

In this section, we first reformulate the optimization problem in \eqref{eq:reformul_opt}, which involves two optimization variables, into a simpler form with only a single optimization variable based on Theorems~\ref{Theorem:K} and~\ref{theorem:tau}. Then, we propose a structure-aware piecewise search algorithm that efficiently solves the reformulated optimization problem.

\subsection{Reformulation of the Optimization Problem}

As discussed in Section \ref{Sec:On-OffScheme}, adopting the on–off scheme in Proposition \ref{prop:power_profile} allows the original problem in \eqref{eq:original_opt} to be reformulated as \eqref{eq:reformul_opt}, reducing the number of optimization variables 
from $M+1$ to two,
$P_{\mathrm a}$ and $\tau$. For a given $P_\mathrm a$, Theorem \ref{Theorem:K} characterizes the minimum number of cooperative users through the condition on \(\Sigma_K\), and Theorem \ref{theorem:tau} determines the corresponding activation threshold. Let \(m_i\) denote the index of the user with the \(i\)-th smallest activation metric \(|h_{m,\mathrm b}|^2/\lambda_{m,\mathrm w}\). Then, under the adopted large-$M$ analytical formulation, the problem in \eqref{eq:reformul_opt} can be reduced to
\begin{equation} \label{eq:master_opt} 
P_{\mathrm a}^{*} = \argmax_{0 \le P_{\mathrm a} \le P_{\max}} R(P_{\mathrm a},K_{\min}(P_{\mathrm a})), 
\end{equation} 
where 
\begin{equation} 
R(P_{\mathrm a},K) = \log_2\left( 1+ \frac{ P_{\mathrm a}|h_{\mathrm a,\mathrm b}|^2 } { \sum_{i=1}^{K}P_{\max}|h_{m_i,\mathrm b}|^2 +\sigma_{\mathrm b}^2 } \right). 
\label{eq:Rate_OnOff} 
\end{equation}
This optimization problem has only one optimization variable, i.e., \(P_{\mathrm a}\), and is therefore much easier to solve than the original formulation of \eqref{eq:original_opt}, which involves \(M+1\) optimization variables.

\subsection{Optimization of Transmit Power for Alice}
Notably, as shown in Theorem~\ref{Theorem:K} and Corollaries~\ref{Corollary:K_uniform} and~\ref{cor:homogeneous}, \(K_{\min}\) is a function of \(P_{\mathrm a}\). In particular, the minimum number of cooperative users increases with \(P_{\mathrm a}\), because a larger \(P_{\mathrm a}\) produces a larger Alice-induced detection shift at Willie and therefore requires stronger cooperative-interference uncertainty to satisfy the covert constraint. Hence, the solution of \eqref{eq:master_opt} is not simply the maximum allowable value \(P_{\max}\). In this work, we determine \(P_{\mathrm a}^{*}\) through a piecewise search for the optimization problem in \eqref{eq:master_opt}.

Meanwhile, it is noteworthy that \(K_{\min}(P_{\mathrm a})\) is an integer-valued function. As a result, \(K_{\min}(P_{\mathrm a})\) remains unchanged over certain ranges of \(P_{\mathrm a}\). 
{\blue Moreover, for a fixed cooperation size $K$, the set of selected cooperative users and the resulting aggregate interference at Bob in the denominator of \eqref{eq:Rate_OnOff} do not vary with $P_{\mathrm{a}}$. Hence, $R(P_{\mathrm{a}},K)$ in \eqref{eq:Rate_OnOff} is strictly increasing in $P_{\mathrm{a}}$ within each interval where $K_{\min}(P_{\mathrm{a}})=K$, and the optimum within each interval is attained at the largest feasible transmit power, i.e., the boundary candidate of that interval.} 
 For a given cooperation size \(K_{\min}=\bar K\), the maximum value of \(P_{\mathrm a}\) that satisfies the covert constraint can be obtained from Theorem~\ref{Theorem:K} as
\begin{equation} 
\label{eq:Pa_of_K_exact_opt} 
P_{\mathrm a}(\bar K) = \frac{\sqrt{\Sigma_{\bar K}}} {\sqrt{c_\epsilon}\lambda_{\mathrm a,\mathrm w}}. 
\end{equation}
Then, for a given \(K_{\min}=\bar K\), the candidate transmit power is defined as
\(P_{\bar K}=\min\big(P_{\max},P_{\mathrm a}(\bar K)\big)\) for
\(\bar K\in\{0,\ldots,M\}\). Thus, when performing a one-dimensional search for
the solution of \eqref{eq:master_opt}, it is sufficient to evaluate the
candidate boundary points \(\{P_{\bar K}\}_{\bar K=0}^{M}\). This piecewise
characterization eliminates the need to explore redundant regions of the search
space, making the resulting search more efficient and structure-aware than an exhaustive one-dimensional search. 
The overall piecewise search is summarized in Algorithm~\ref{alg:piecewise}.

\begin{algorithm}[t] 
\caption{Piecewise Search}
\label{alg:piecewise}
\ifCLASSOPTIONonecolumn\small\fi
\begin{algorithmic}[1] 
\renewcommand{\algorithmicrequire}{\textbf{Input:}}
\renewcommand{\algorithmicensure}{\textbf{Output:}}
\REQUIRE \(h_{\mathrm a,\mathrm b}\), \(\{h_{m,\mathrm b}\}_{m=1}^{M}\), \(\{\lambda_{m,\mathrm w}\}_{m=1}^{M}\), \(P_{\max}\), \(\epsilon\), \(\sigma_{\mathrm b}^{2}\)
\ENSURE \(P_{\mathrm a}^{*}\), \(K^{*}\), \(\tau^{*}\)
\STATE Sort users in ascending order of \(|h_{m,\mathrm b}|^{2}/\lambda_{m,\mathrm w}\), and denote the ordered indices by \(m_1,\ldots,m_M\).

\STATE Compute \(c_\epsilon\) and \(\{\Sigma_{\bar K}\}_{\bar K=0}^{M}\), and set $R_{\text{temp}} \leftarrow 0$.
\FOR{$\bar K=0$ to $M$}
  \STATE Compute \(P_{\bar K}\) from \eqref{eq:Pa_of_K_exact_opt} and \(R(P_{\bar K},\bar K)\) in \eqref{eq:Rate_OnOff}.
  \STATE If \(R(P_{\bar K},\bar K)>R_{\text{temp}}\), update \(R_{\text{temp}}\leftarrow R(P_{\bar K},\bar K)\), \(P_{\mathrm a}^{*}\leftarrow P_{\bar K}\), and \(K^{*}\leftarrow \bar K\).
\ENDFOR
\STATE Set \(\tau^{*}\leftarrow0\) if \(K^{*}=0\); \\
otherwise, \(\tau^{*}\leftarrow |h_{m_{K^{*}},\mathrm b}|^{2}/ \lambda_{m_{K^{*}},\mathrm w}\).
\STATE Return \(P_{\mathrm a}^{*}\), \(K^{*}\), \(\tau^{*}\).
\end{algorithmic}
\end{algorithm}

{\blue The computational complexity of Algorithm~\ref{alg:piecewise} can be explained
directly from its steps. First, computing the activation metric
\(|h_{m,\mathrm b}|^{2}/\lambda_{m,\mathrm w}\) for all \(M\) users requires
\(\mathcal O(M)\) operations, and sorting the users in ascending order of this
metric requires \(\mathcal O(M\log M)\) operations. Then, the main loop
evaluates one candidate transmit power \(P_{\bar K}\) and one corresponding
rate \(R(P_{\mathrm a},\bar K)\) for each \(\bar K=0,\ldots,M\), which requires
\(\mathcal O(M)\) candidate evaluations. Hence, the overall complexity of the
piecewise search is $ \mathcal O(M)+\mathcal O(M\log M)+\mathcal O(M) = \mathcal O(M\log M).$}
In contrast, a grid-based one-dimensional search over \(P_{\mathrm a}\) with
\(N_p\) grid points requires \(\mathcal O(N_pM)\) operations, because the
corresponding number of cooperative users and achievable rate must be evaluated for each
grid point. Moreover, its solution depends on the power-grid resolution and may
miss the analytical boundary candidate \(P_{\bar K}\). 
{\blue The proposed piecewise search instead evaluates all analytically relevant candidates \(P_{\bar K}\), thereby reducing the computational burden and avoiding power-discretization error.
Table~\ref{tab:grid_error_comparison} quantifies this finite-resolution effect.
The grid-based search selects slightly different values of \(P_{\mathrm a}^{\star}\) and \(K^\star\), which leads to a small rate loss, whereas Algorithm~\ref{alg:piecewise} directly evaluates the boundary candidates \(P_{\bar K}\) for all \(\bar K=0,\ldots,M\).}

\begin{table}[t]
\centering
\caption{ Optimization comparison between Algorithm~1 and the grid-based search with
$N_P=10^4$. The reported values are averaged over independent fading realizations, and
the rate gain is defined as $\left(R^\star-R^\star_{\mathrm{Grid}}\right)/R^\star_{\mathrm{Grid}}$.}
\label{tab:grid_error_comparison}
\renewcommand{\arraystretch}{1.12}
\scriptsize
\begin{tabular}{c c c c c c}
\hline
\(M\)
&
\(P_{\mathrm a}^\star/P_{\max}\)
&
\(P_{\mathrm a,\mathrm{Grid}}^\star/P_{\max}\)
&
\(K^\star\)
&
\(K_{\mathrm{Grid}}^\star\)
&
Rate gain \\
\hline
500  & 0.336 & 0.319 & 73.67  & 74.53  & 5.70\% \\
1000 & 0.407 & 0.386 & 103.29 & 104.15 & 5.94\% \\
2000 & 0.486 & 0.458 & 144.35 & 145.07 & 6.22\% \\
\hline
\end{tabular}
\end{table}

When imperfect user-to-Bob channel estimation is considered, the same piecewise structure can be applied by replacing \(h_{m,\mathrm b}\) in the activation metric with its estimate \(\hat h_{m,\mathrm b}\). Under the adopted channel estimation error model, \(\hat h_{m,\mathrm b}\) has the same statistical distribution as \(h_{m,\mathrm b}\), and hence Willie-side aggregate interference statistics remain unchanged for the same \((P_{\mathrm a},K)\). Therefore, the channel estimation error mainly affects the set of cooperative users $\mathcal{\hat S} \neq \mathcal{S}$ and the resulting achievable rate at Bob, while the Willie-side DEP retains the same performance.

\section{Numerical Results} \label{Sec:Simulation}

In this section, we numerically validate the theoretical analyses and examine the interactions among key system parameters in the proposed covert communication network with spatially distributed users.
To reflect more general wireless deployments, the numerical results are obtained under randomized spatial configurations rather than a fixed adverse topology.
For performance evaluations, we consider the large-scale fading channel model specified in the 3GPP TR 25.996 Urban Macrocell guideline~\cite{Ref_3GPP_TR25996}. The path loss (PL) at a distance $d$ (in meters) is modeled as
\begin{equation}
    \text{PL(dB)} = 34.5 + 35\log_{10}(d),
\end{equation}
which corresponds to the large-scale fading coefficient $\lambda_{x,y} = \beta_0 d_{x,y}^{-\alpha}$. Here, the path-loss exponent is set to $\alpha=3.5$, and the reference channel gain at $d=1\mathrm{m}$ is given by $\beta_0=10^{-34.5/10}$. The maximum transmit power of each user is set to $P_{\max} = 200\,\mathrm{mW}\,(23\,\mathrm{dBm})$, and the noise variances at Bob and Willie are assumed to be identical, i.e., $\sigma_{\mathrm b}^2 = \sigma_{\mathrm w}^2 = -102\,\mathrm{dBm}$.
The spatial topology of the network is constructed over a $1000\,\mathrm{m}\times1000\,\mathrm{m}$ square region. Willie is placed at the center of the area, $(500,500)$, while Bob is located at $(100,100)$. Alice is located at $(832.3,832.3)$, which gives $d_{\mathrm a,\mathrm w}\approx 470\,\mathrm{m}$ and $d_{\mathrm a,\mathrm b}\approx 1035.7\,\mathrm{m}$. This placement maintains a relatively unfavorable Alice-to-Bob link compared with the Alice-to-Willie link. The non-covert users are randomly distributed over the simulation area. This randomized topology allows us to examine the proposed scheme under spatially heterogeneous user deployments without restricting the evaluation to a specific geometry.

\begin{figure}[t]
\centering
\includegraphics[width=\figwidth]{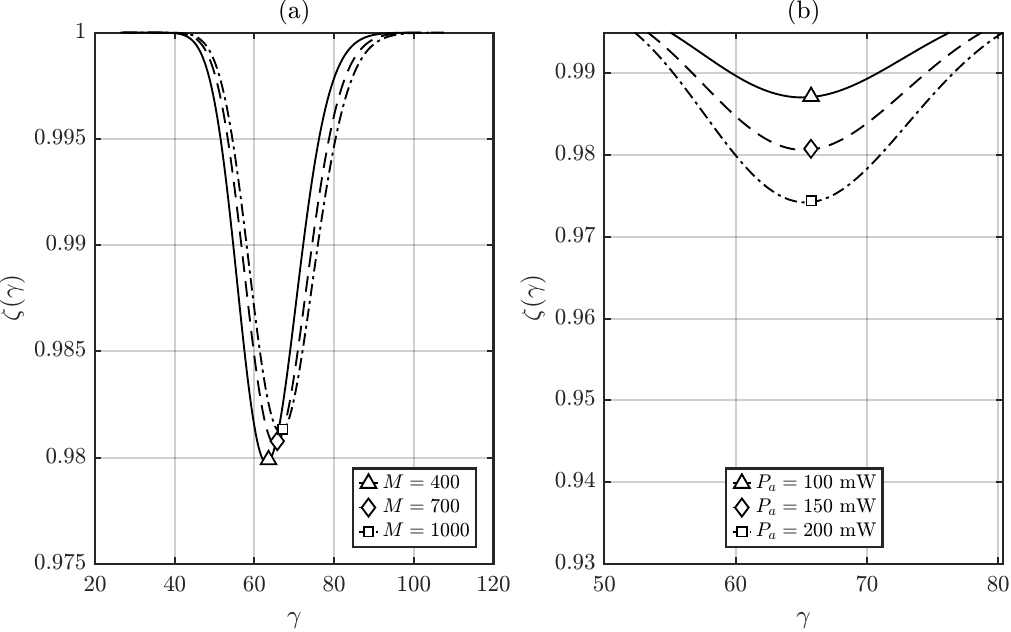}
\caption{DEP $\zeta(\gamma)$ as a function of the detection threshold
$\gamma$ under the general randomized topology. In (a), $K=100$ and
$P_{\mathrm a}=150\,\mathrm{mW}$ are fixed, while $M$ varies. In (b),
$M=700$ and $K=100$ are fixed, while $P_{\mathrm a}$ varies. The solid,
dashed, and dash-dot lines denote the numerical DEP, and the markers indicate the analytical
results of $\gamma^\star$ and $\zeta_{\min}$ obtained from Lemmas~2 and~3.}
\label{fig:DEP_gamma}
\end{figure}

In Fig.~\ref{fig:DEP_gamma}, we present the DEP curves obtained from numerical simulations under the general randomized topology. The line curves represent the DEP $\zeta(\gamma)$ as a function of Willie’s detection threshold $\gamma$, while the markers correspond to the theoretical results in Lemmas~\ref{lemma:gamma} and~\ref{lemma:Min_DEP}. Specifically, the $x$-coordinates of the markers represent the optimal detection threshold, $\gamma^\star$, derived in Lemma~\ref{lemma:gamma}, and the $y$-coordinates represent the minimum DEP, $\zeta_{\min}$, obtained from Lemma~\ref{lemma:Min_DEP}. In Fig.~\ref{fig:DEP_gamma}(a), we consider different numbers of non-covert users, i.e., $M=400$, $700$, and $1000$, while fixing $K=100$ and $P_{\mathrm a}=150\,\mathrm{mW}$. In Fig.~\ref{fig:DEP_gamma}(b), we consider different transmit powers of Alice, i.e., $P_{\mathrm a}=100$, $150$, and $200\,\mathrm{mW}$, while fixing $M=700$ and $K=100$.
As observed in Fig.~\ref{fig:DEP_gamma}, the theoretical results for $\gamma^\star$ and $\zeta_{\min}$ closely match the minima of the corresponding DEP curves for all considered system parameters. This close agreement validates Lemmas~\ref{lemma:DEP}--\ref{lemma:Min_DEP}.

It is also observed in Fig.~\ref{fig:DEP_gamma}(a) that increasing the number of non-covert users $M$ slightly increases the minimum DEP under the fixed values of $K$ and $P_{\mathrm a}$. This behavior can be explained by the order-statistics-induced cooperative user set. When $M$ increases while $K$ is fixed, Bob selects the $K$ cooperative users from a larger candidate pool according to the activation metric $r_m$. Hence, the selected cooperative users are more likely to provide effective Willie-side interference uncertainty relative to their Bob-side interference cost. As a result, the variance of the aggregate cooperative interference observed at Willie increases relative to the Alice-induced shift $\Delta=P_{\mathrm a}\lambda_{\mathrm a,\mathrm w}$. Since the minimum DEP is governed by the ratio between the aggregate interference uncertainty and the Alice-induced shift, the distributions of test statistics under $\mathcal H_0$ and $\mathcal H_1$ become more overlapped, which increases Willie’s DEP.
Meanwhile, Fig.~\ref{fig:DEP_gamma}(b) shows that increasing Alice’s transmit power from $P_{\mathrm a}=100\,\mathrm{mW}$ to $200\,\mathrm{mW}$ reduces the DEP. This occurs because a higher $P_{\mathrm a}$ increases the Alice-induced received-energy shift at Willie, making the two hypotheses more distinguishable and thereby improving Willie’s detection capability. Therefore, a larger Alice transmit power requires stronger cooperative interference uncertainty to maintain the same level of covertness.

\begin{figure}[t]
\centering
\includegraphics[width=\figwidth]{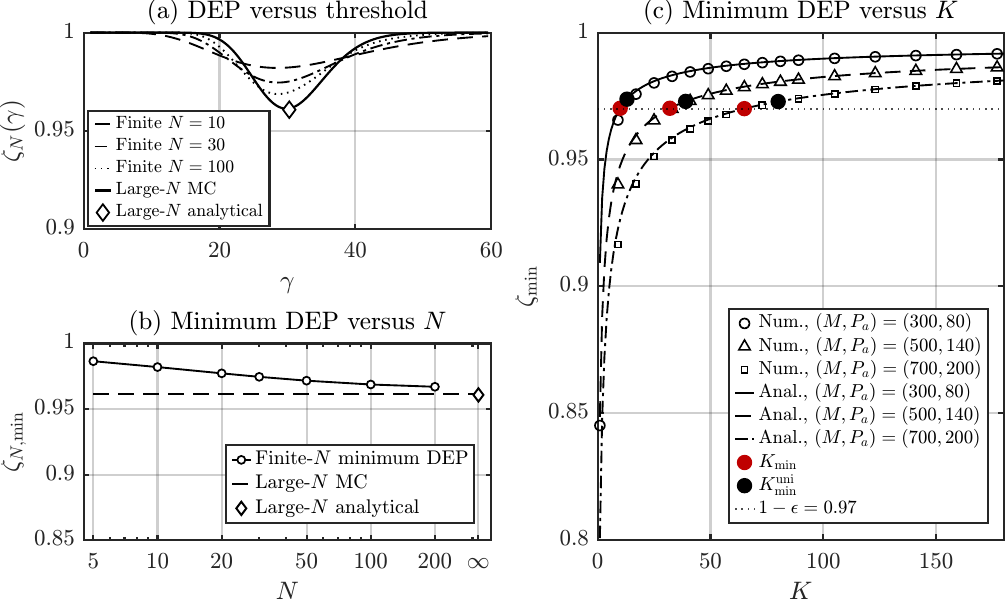}
\caption{Finite-$N$ DEP behavior and minimum DEP versus the number of cooperative users. In (a) and (b), $M=700$, $K=50$, and $P_{\mathrm a}=200\,\mathrm{mW}$ are fixed. In (c), $\zeta_{\min}$ is plotted as a function of $K$ for $(M,P_{\mathrm a})=(300,80\,\mathrm{mW})$, $(500,140\,\mathrm{mW})$, and $(700,200\,\mathrm{mW})$. The dotted horizontal line represents the covert constraint $1-\epsilon=0.97$.}
\label{fig:MinDEP_vs_K}
\end{figure}

In Fig.~\ref{fig:MinDEP_vs_K}, we further examine Willie’s detection performance and validate the theoretical results presented in Lemmas~\ref{lemma:DEP}, \ref{lemma:gamma}, \ref{lemma:Min_DEP}, and Theorem~\ref{Theorem:K}. In particular, Figs.~\ref{fig:MinDEP_vs_K}(a) and~(b) evaluate the finite-$N$ behavior of Willie’s energy detector, while Fig.~\ref{fig:MinDEP_vs_K}(c) validates the minimum DEP, $\zeta_{\min}$, and the minimum number of cooperative users, $K_{\min}$, under the proposed on--off scheme in Proposition~\ref{prop:power_profile}.
In Fig.~\ref{fig:MinDEP_vs_K}(a), we plot the DEP, $\zeta_N(\gamma)$, as a function of the detection threshold, $\gamma$, for different sample lengths, i.e., $N=10$, $30$, and $100$. The large-$N$ curve and the analytical marker are also included as benchmarks. As shown in the figure, the finite-$N$ DEP curves approach the large-$N$ benchmark as $N$ increases. This is because, for finite $N$, Willie’s received-energy statistic still contains finite-sample fluctuation, whereas in the large-$N$ limit, this fluctuation vanishes and the test statistic converges to the instantaneous aggregate received power. Therefore, the large-$N$ analysis provides a favorable benchmark for Willie’s detection capability. Fig.~\ref{fig:MinDEP_vs_K}(b) further confirms this observation by showing that the finite-$N$ minimum DEP, $\zeta_{N,\min}$, gradually converges to the large-$N$ result as $N$ increases.

In Fig.~\ref{fig:MinDEP_vs_K}(c), we plot the minimum DEP, $\zeta_{\min}$, as a function of the number of cooperative users, $K$, under the general randomized topology. The numerical results are determined by minimizing the DEP over Willie’s detection threshold, $\gamma$, for each value of $K$. The analytical curves correspond to the minimum DEP derived in Lemma~\ref{lemma:Min_DEP}. As observed in Fig.~\ref{fig:MinDEP_vs_K}(c), the analytical results closely match the numerical results for all considered parameter settings. This agreement confirms that Lemma~\ref{lemma:Min_DEP} accurately characterizes the minimum DEP and that Theorem~\ref{Theorem:K} serves as a reliable and computationally efficient guideline to determine the necessary number of cooperative users to meet a specified covert constraint.
{\blue Furthermore, the red and black markers in Fig.~\ref{fig:MinDEP_vs_K}(c) indicate $K_{\min}$ obtained from Theorem~\ref{Theorem:K} and $K_{\min}^{\mathrm{uni}}$ obtained from Corollary~\ref{Corollary:K_uniform}, respectively. The comparison between these markers shows that the minimum number of cooperative users depends not only on the total number of available users and Alice’s transmit power, but also on how the cooperative users are selected. In particular, Theorem~\ref{Theorem:K} accounts for the order-statistics-induced distribution of the set of cooperative users, whereas the uniform approximation neglects this selection bias. This explains the gap between $K_{\min}$ and $K_{\min}^{\mathrm{uni}}$.}
{\blue It is also observed in Fig.~\ref{fig:MinDEP_vs_K}(c) that $\zeta_{\min}$ increases with $K$ and gradually enters a stable region. As $K$ increases, $\Sigma_K$ increases, and hence Alice's energy shift becomes more difficult to distinguish from the cooperative interference fluctuation. Accordingly, Willie’s DEP increases. However, when $K$ is sufficiently large, i.e., $\Sigma_K\gg\Delta^2$, the remaining gap satisfies
\begin{equation}\label{eq:DEP_gap_asymp}
    1-\zeta_{\min}
    \approx
    \frac{\Delta}{\sqrt{2\pi\Sigma_K}}.
\end{equation}
Therefore, even though increasing $K$ enlarges $\Sigma_K$, the remaining 
gap $1-\zeta_{\min}$ decreases only with the inverse square root of 
$\Sigma_K$. This explains why $\zeta_{\min}$ first increases rapidly and 
then becomes nearly stable. Thus, after the covert constraint is 
satisfied, activating additional cooperative users only slightly increases 
$\zeta_{\min}$, while it continuously increases the interference imposed on Bob. }

\begin{figure}[t]
\centering
\includegraphics[width=\figwidth]{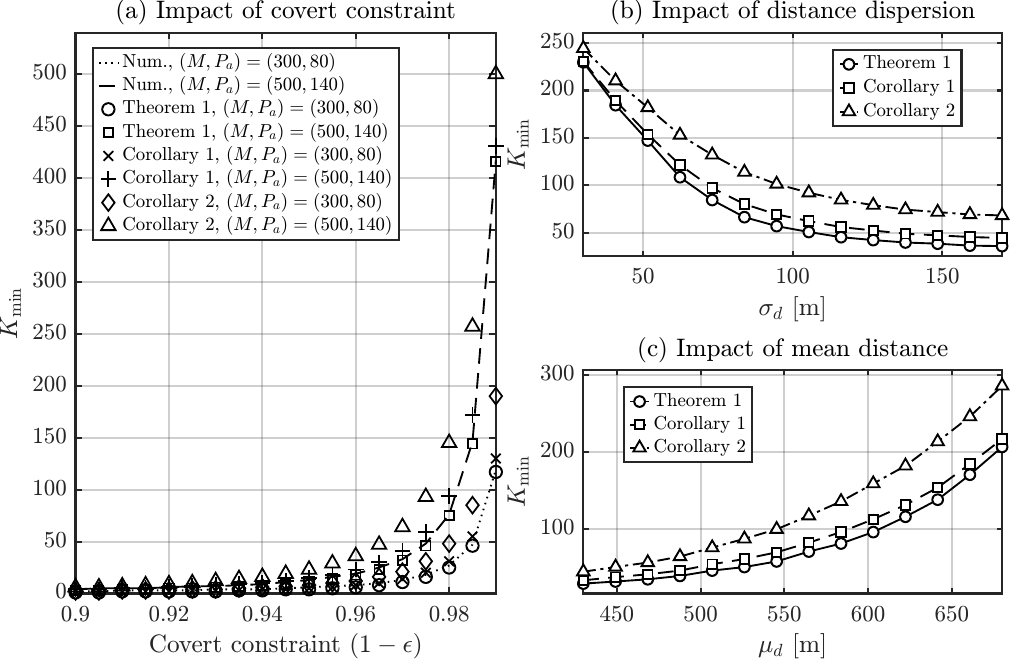}
\caption{Impact of the covert constraint and spatial parameters on the minimum number of cooperative users, $K_{\min}$. In (a), $K_{\min}$ is plotted versus the covert constraint $1-\epsilon$ for $(M,P_{\mathrm a})=(300,80\,\mathrm{mW})$ and $(500,140\,\mathrm{mW})$. In (b) and (c), $K_{\min}$ is plotted versus the standard deviation $\sigma_d$ and the mean $\mu_d$ of the user-to-Willie distances, respectively, for $M=500$, $P_{\mathrm a}=140\,\mathrm{mW}$, and $1-\epsilon=0.97$.}
\label{fig:epsilon_Kmin}
\end{figure}

In Fig.~\ref{fig:epsilon_Kmin}(a), we plot $K_{\min}$ as a function of the covert constraint, $1-\epsilon$, under the general randomized topology. The line curves represent the numerical values of $K_{\min}$ while the markers indicate the analytical results derived from Theorem~\ref{Theorem:K}, Corollary~\ref{Corollary:K_uniform}, and Corollary~\ref{cor:homogeneous}. Fig.~\ref{fig:epsilon_Kmin}(a) illustrates that the results of Theorem~\ref{Theorem:K} closely match the numerical results for both considered network sizes, i.e., $M=300$ and $M=500$. This confirms that the proposed analysis of $K_{\min}$ remains accurate not only in a large-user regime, but also for smaller and moderate user populations. In addition, as the covert constraint becomes more stringent, i.e., as $1-\epsilon$ increases, $K_{\min}$ increases. This is because a stricter covert constraint requires larger aggregate cooperative interference uncertainty at Willie.
The results of Corollary~\ref{Corollary:K_uniform} follow the same trend as those of Theorem~\ref{Theorem:K}, which confirms that the uniform finite-population approximation provides a useful simplified insight. However, the results of Corollary~\ref{cor:homogeneous} are noticeably larger than those of Theorem~\ref{Theorem:K}. This implies that the homogeneous approximation can overestimate the minimum number of cooperative users when the actual user deployment is spatially heterogeneous. 

Figs.~\ref{fig:epsilon_Kmin}(b) and (c) further illustrate how the spatial distribution of non-covert users affects $K_{\min}$. In Fig.~\ref{fig:epsilon_Kmin}(b), the standard deviation of the user-to-Willie distances, $\sigma_d$, is varied. It is observed that as $\sigma_d$ increases, the minimum number of cooperative users decreases. This is because a larger distance dispersion creates a more heterogeneous set of user-to-Willie links, which increases the chance that selected cooperative users provide stronger interference uncertainty at Willie. As a result, fewer cooperative users are required to satisfy the same covert constraint.
In Fig.~\ref{fig:epsilon_Kmin}(c), the mean user-to-Willie distance, $\mu_d$, is varied. It is observed that as $\mu_d$ increases, $K_{\min}$ also increases. This is because the users become farther from Willie on average, and the aggregate interference power observed at Willie becomes weaker. Hence, more cooperative users are required to generate sufficient interference uncertainty and satisfy the covert constraint. Combining the results from Figs. \ref{fig:epsilon_Kmin} (a) and (b) yields a comprehensive design guideline: to maximize covert communication performance, non-covert users should be deployed as close to Willie as possible while maintaining a high degree of spatial dispersion.

\begin{figure}[t]
\centering
\includegraphics[width=\figwidth]{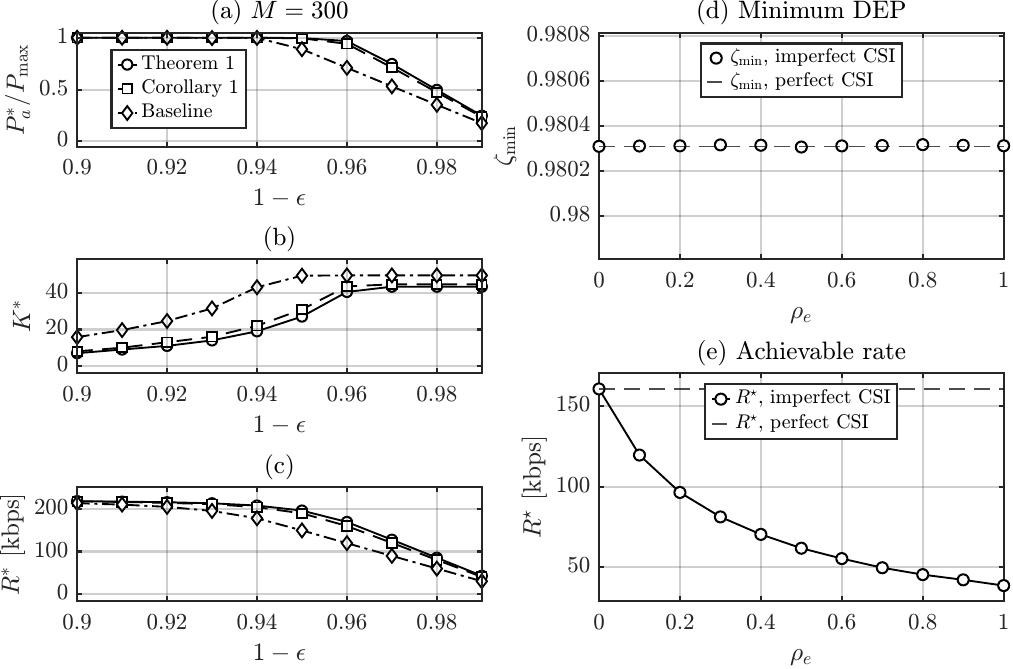}
\caption{Optimized system performance and impact of imperfect user-to-Bob channel estimation. In (a)--(c), $P_{\mathrm a}^{\star}/P_{\max}$, $K^{\star}$, and $R^{\star}$ are plotted versus the covert constraint $1-\epsilon$ for $M=300$. In (d), Willie’s minimum DEP is plotted versus the channel estimation error level $\rho_e$ for $M=700$, $K=100$, and $P_{\mathrm a}=150\,\mathrm{mW}$. In (e), the optimized achievable rate is plotted versus $\rho_e$ for $M=700$.}
\label{fig:optimal}
\end{figure}

In Fig.~\ref{fig:optimal}, we compare the covert communication performance of the proposed design with those of the uniform-approximation design and the Bob-only baseline. In the proposed design, cooperative users are selected based on Proposition~\ref{prop:power_profile}. 
{\blue In contrast, the Bob-only baseline selects cooperative users solely according to the instantaneous channel gains between the users and Bob~\cite{Lee23Multi,Lee24Channel,Yeom24Covert}. 
Importantly, both the proposed method and the baseline method are evaluated under the same heterogeneous network realization. The baseline differs from the proposed method only in the cooperative user selection rule. It selects cooperative users based on the instantaneous user-to-Bob channel gains only, without using the user-to-Willie large-scale fading coefficients in the selection metric. After the baseline user set is determined, its covertness and achievable rate are evaluated using the actual heterogeneous user-to-Willie large-scale fading coefficients. }
To validate the effectiveness of the proposed design, we compare the optimized results, i.e., $P_{\mathrm a}^{\star}$, $K^{\star}$, and $R^{\star}$, with respect to the stringency of the covert constraint, $1-\epsilon$.

In Figs.~\ref{fig:optimal}(a) and~(b), we present the ratio of the optimized transmit power for the covert message to the maximum transmit power, $P_{\mathrm a}^{\star}/P_{\max}$, and the optimal number of cooperative users, $K^{\star}$, respectively. It is observed that all schemes adjust $P_{\mathrm a}^{\star}$ and $K^{\star}$ according to the stringency of the covert constraint. However, the proposed design requires fewer cooperative users than the Bob-only baseline for a given covert constraint. The proposed design exploits this spatial heterogeneity by selectively utilizing users that can effectively increase Willie’s uncertainty without excessively increasing the interference at Bob. In contrast, the Bob-only baseline does not exploit the user-to-Willie large-scale fading coefficients in the selection metric, and therefore cannot fully utilize the spatially heterogeneous interference structure. 
In addition, Fig.~\ref{fig:optimal}(c) demonstrates that the proposed design improves the optimized covert rate compared with the Bob-only baseline. In particular, as the covert constraint becomes more stringent, the performance gap between the proposed design and the baseline becomes more pronounced. This indicates that incorporating the distinct spatial locations of users into the cooperative user selection and optimization process substantially enhances covert performance. The performance of Corollary~\ref{Corollary:K_uniform} is close to that of Theorem~\ref{Theorem:K}, which confirms that the uniform approximation provides a useful simplified design. 

Fig.~\ref{fig:optimal}(d) shows the impact of imperfect user-to-Bob channel estimation on Willie’s minimum DEP at $M=700$, $K=100$, and $P_{\mathrm a}=150\,\mathrm{mW}$. {\blue Under the adopted error model,
$\hat h_{m,\mathrm b}=\sqrt{1-\rho_e}h_{m,\mathrm b}+\sqrt{\rho_e}u_{m,\mathrm b}$,
where $u_{m,\mathrm b}$ has the same distribution as $h_{m,\mathrm b}$. Thus, changing $\rho_e$ changes how accurately $\hat h_{m,\mathrm b}$ represents the true channel $h_{m,\mathrm b}$, but does not change the statistical distribution of $\hat h_{m,\mathrm b}$ itself. Therefore, for the same $(P_{\mathrm a},K)$, the Willie-side aggregate interference statistics remain unchanged, and the minimum DEP overlaps with the perfect estimation reference. }
Furthermore, Fig.~\ref{fig:optimal}(e) shows the impact of imperfect user-to-Bob channel estimation on the optimized achievable rate. {\blue Bob selects cooperative users and determines $K^{\star}$, the activation threshold, and $P_{\mathrm a}^{\star}$ using the estimated channels $\hat h_{m,\mathrm b}$. When the estimation error increases, users selected as weak interferers to Bob based on $\hat h_{m,\mathrm b}$ can have large true channel gains $h_{m,\mathrm b}$. Therefore, the actual aggregate interference at Bob increases, thereby reducing $R^{\star}$.}

\section{Conclusion}\label{Sec:conclusion}

This paper investigated covert communication in which a single covert user is assisted by multiple spatially distributed non-covert users against a geometry-aware Willie. Under the adopted large-$M$ analytical formulation, we established an on--off cooperative-user activation structure and characterized Willie’s DEP, the minimum number of cooperative users, and the corresponding activation threshold. We further developed a low-complexity piecewise search algorithm for the analytically reformulated power-control problem. The results showed that greater Willie-side spatial dispersion can reduce the minimum number of cooperative users, whereas a larger average user-to-Willie distance increases it, and that homogeneous or averaged-geometry analyses may mischaracterize the cooperation requirement. The numerical results further confirmed that the finite-sample DEP converges to the large-sample benchmark with an increasing number of observations, and that the impact of imperfect user-to-Bob channel estimation appears primarily as a rate loss induced by cooperative user activation errors, rather than as a loss of covertness.
{\blue Future work will extend the proposed framework to mobile networks, multiple-Willie scenarios, and integrated sensing and communication systems, where sensing resources and spatial information can be jointly exploited to enhance covertness.}

\appendices

\section{Proof of Proposition \ref{prop:power_profile}} \label{Appendix:power_profile}

For a fixed $P_a$, \eqref{eq:original_opt} becomes
\begin{subequations}\label{eq:powerprofile_Opt}
\begin{align}
  \bar{P}^* & = \argmax_{\bar{P}} \frac{P_\mathrm a |h_\mathrm{a,b}|^2}{\sum_{m=1}^M P_m|h_{m,\mathrm b}|^2 + \sigma_\mathrm b^2} \\
  & \text{s.t.} ~~ \eqref{eq:CovertConst} ~ \text{and} ~ \eqref{eq:PowConst_Noncovert}.\nonumber
\end{align}
\end{subequations}

To characterize the covert constraint, we approximate the sum of Gamma random variables by a Gamma distribution with matched first and second moments~\cite{Mathai93AHandbook,Moschopoulos85TheDistribution,Covo14AnovelSingleGammaApprox}.  Thus, in large $M$, $T_\mathrm w$ can be reformulated to 
\begin{equation} \label{eq:pdfTest}
T_\mathrm w - \sigma_\mathrm w^2 \sim \begin{cases}
    \Gamma\!\left( \frac{\psi^2}{\Psi},\,\frac{\Psi}{\psi} \right), & \mathcal{H}_0, \\[3pt]
    \Gamma\!\left( \frac{(\psi + \lambda_\mathrm{a,w}P_\mathrm a)^2}{\Psi + (\lambda_\mathrm{a,w} P_\mathrm a)^2},\,\frac{\Psi + (\lambda_\mathrm{a,w} P_\mathrm a)^2}{\psi + \lambda_\mathrm{a,w}P_\mathrm a} \right), & \mathcal{H}_1,
  \end{cases}
\end{equation}
where $\psi = \sum_{m=1}^M \lambda_{m,\mathrm w}P_m$ and $\Psi = \sum_{m=1}^M \lambda_{m,\mathrm w}^2P_m^2$.
For sufficiently large $M$, this Gamma distribution is approximated as a normal distribution~\cite[Sec. 5.5]{Casella02StatInference}, \cite[Sec. 2.8]{Ross19Probability}, yielding
\begin{equation} \label{eq:ProbDist_Tw}
    T_\mathrm w - \sigma_\mathrm w^2 \sim \begin{cases}
      \mathcal{N}\!\left(\psi,\,\Psi \right), & \mathcal{H}_0, \\[3pt]
      \mathcal{N}\!\left(\psi + \lambda_{a,w}P_\mathrm a,\, \Psi + \lambda^2_\mathrm{a,w}P^2_\mathrm a \right), & \mathcal{H}_1.
\end{cases}
\end{equation}
Using \eqref{eq:ProbDist_Tw}, the DEP is given by
\begin{align}
\zeta &= \Pr(T_\mathrm w > \gamma \mid \mathcal{H}_0) + \Pr(T_\mathrm w \leq \gamma \mid \mathcal{H}_1) \nonumber \\
&= \underbrace{Q\!\left(\frac{\gamma-\psi}{\sqrt{\Psi}}\right)}_{(a)} 
   + \underbrace{Q\!\left(\frac{(\psi + \lambda_\mathrm{a,w}P_\mathrm a) - \gamma}{\sqrt{\Psi + \lambda_\mathrm{a,w}^2P_\mathrm a^2}}\right)}_{(b)}, \label{eq:DEP_Gaussian}
\end{align}
where $Q(\cdot)$ denotes the standard Gaussian $Q$-function. Setting $P_{\rm FA}=\alpha$, the detection threshold is obtained from term (a) as
\begin{equation}
    \gamma = \sqrt{\Psi}\,Q^{-1}(\alpha) + \psi. \label{eq:FAwithAlpha}
\end{equation}

Substituting \eqref{eq:FAwithAlpha} into term (b) yields \(P_{\rm MD}\), and thus
\begin{equation}\label{eq:DEP_Pi}
\zeta = \alpha + Q\!\left(\frac{\lambda_\mathrm{a,w}P_\mathrm a - \sqrt{\Psi}\,Q^{-1}(\alpha)}{\sqrt{\Psi + \lambda_\mathrm{a,w}^2P_\mathrm a^2}}\right).
\end{equation}
From \eqref{eq:DEP_Pi}, $\zeta(\alpha,\Psi)$ is monotonically increasing in $\Psi$ for any fixed $\alpha$. Hence, the minimized DEP $\zeta_{\min}(\Psi)=\min_\alpha \zeta(\alpha,\Psi)$ is also monotonically increasing in $\Psi$.
Suppose that for two values $\Psi < \hat\Psi$, there exist optimal thresholds $\alpha^*$ and $\hat\alpha$ minimizing $\zeta(\alpha,\Psi)$ and $\zeta(\alpha,\hat\Psi)$, respectively. 

Then, by definition, $\zeta_{\min}(\Psi) = \zeta(\alpha^*,\Psi) \leq \zeta(\hat\alpha,\Psi)$, and because $\zeta$ increases with $\Psi$ for any fixed $\alpha$, $\zeta(\hat\alpha,\Psi) < \zeta(\hat\alpha,\hat\Psi) = \zeta_{\min}(\hat\Psi)$.
Therefore, $\zeta_{\min}(\Psi)$ is strictly increasing with respect to $\Psi$. By applying Sedrakyan’s Lemma~\cite{Sedrakyan97AboutTheApplications}, we obtain
\begin{equation} \label{eq:Cauchy}
  \frac{\Psi}{M}
  = \frac{1}{M}\sum_{m=1}^M \!\left(\lambda_{m,\mathrm w}P_m\right)^2
  \ge \frac{1}{M^2}\!\left(\sum_{m=1}^M \lambda_{m,\mathrm w}P_m\right)^2,
\end{equation}
where the inequality becomes tight as $M$ grows large. Hence, for sufficiently large $M$, we have $\Psi \approx \psi^2/M$, and $\zeta_{\min}$ becomes a monotonically increasing function of $\psi^2/M$. Therefore, the covert constraint becomes
\begin{equation}\label{eq:convert_DEP_Pi}
  \psi = \sum_{m=1}^M \lambda_{m,\mathrm w}P_m \geq \delta,
\end{equation}
where $\delta \triangleq \sqrt{M}\,\zeta_{\min}^{-1}(1-\epsilon)$.

The optimization problem in \eqref{eq:powerprofile_Opt} can be reformulated as
\begin{subequations}\label{eq:powerprofile_Opt_reformul}
\begin{align}
  \bar{P}^* = \arg\min_{\bar{P}} & \sum_{m=1}^M P_m|h_{m,\mathrm b}|^2 + \sigma_\mathrm b^2,  \label{eq:powerprofile_Opt_reformul_obj}\\
  & \text{s.t.} ~~ \sum_{m=1}^M \lambda_{m,\mathrm w}P_m \ge \delta ~~ \text{and} ~~ \eqref{eq:PowConst_Noncovert}     \label{eq:CovertConst_IntfVer}
\end{align}
\end{subequations}
The KKT conditions with multipliers $\{\phi_m,\omega_m,\nu\}$ are
\begin{itemize}
  \item \textbf{Stationarity:} 
  $|h_{m,b}|^2 - \nu^*\lambda_{m,w} - \phi_m^* + \omega_m^* = 0$, $\forall m$.
  \item \textbf{Complementary slackness:} 
  $\phi_m^* P_m^* = 0$, 
  $\omega_m^*(P_m^* - P_{\max}) = 0$, and 
  $\nu^* \left(\delta-\sum_{m=1}^{M} \lambda_{m\mathrm w}P_m^*\right) = 0$.
  \item \textbf{Feasibility:} 
  $0 \leq P_m^* \leq P_{\max}$, 
  $\sum_{m=1}^M \lambda_{m,\mathrm w}P_m^* \ge \delta$, and 
  $\phi_m^*, \omega_m^*, \nu^* \ge 0$.
\end{itemize}
Solving these conditions for each $m$ yields three feasible cases according to the comparison between $|h_{m,\mathrm b}|^2/\lambda_{m,\mathrm w}$ and $\nu^*$.

Thus, the optimal transmit power for each non-covert user is expressed as
\begin{equation}
    P_m^* =
\begin{cases}
      P_{\max}, & \frac{|h_{m,\mathrm b}|^2}{\lambda_{m,\mathrm w}} < \tau, \\[3pt]
      0, & \frac{|h_{m,\mathrm b}|^2}{\lambda_{m,\mathrm w}} > \tau, \\[3pt]
      \displaystyle \frac{\delta - \sum_{ m \in \mathcal{C} } P_{\max}}{\lambda_{m,\mathrm w}}, & \frac{|h_{m,\mathrm b}|^2}{\lambda_{m,\mathrm w}} = \tau,
\end{cases} \label{eq:OptPowProfile}
\end{equation}
where $\mathcal{C}$ denotes the set of selected cooperative users, i.e., $\mathcal{C} = \{m': |h_{m',\mathrm b}|^2 / \lambda_{m',\mathrm w} < \tau \}$ and $\tau = \nu^*$. 
Note that if $U_m$ is a non-covert user satisfying $|h_{m,\mathrm b}|^2 / \lambda_{m,\mathrm w} = \tau$,  the user cannot determine the transmit power from \eqref{eq:OptPowProfile} since $U_m$ cannot have the knowledge of the set $\mathcal{C}$, i.e., the channels between the other users and Bob. 

Consequently, absorbing the boundary case into the activated set, the optimal power profile reduces to the on--off scheme in \eqref{eq:onoff_prop}.
\qed

\section{Proof of Lemma \ref{lemma:DEP}}\label{Appendix:DEP}

By leveraging the moment-matching based Gamma approximation technique introduced in \cite{Mathai93AHandbook, Moschopoulos85TheDistribution, Covo14AnovelSingleGammaApprox}, the PDF of $T_\mathrm w - \sigma_\mathrm w^2$ in \eqref{eq:test_statistic_RV} can be approximated as
\begin{equation} \label{eq:test_statistic_RV_GammaApprox}
   T_\mathrm w - \sigma_\mathrm w^2  \\[-2pt]
      \sim 
     \begin{cases}
      \displaystyle \Gamma\!\left( \frac{\theta^2}{\Theta},\,\frac{\Theta}{\theta} \right), 
      \hfill \mathcal{H}_0, \\
      \displaystyle \Gamma\!\left( 
        \frac{\theta^2}{\Theta},
        \frac{\Theta }{\theta}\right) \ast \Gamma\left(1, \lambda_\mathrm{a,w}P_\mathrm a
      \right), 
      \hfill \mathcal{H}_1,
     \end{cases} 
\end{equation}
where $\theta = \sum_{i=1}^K \lambda_{m_i,\mathrm w}P_{\max}$ and $\Theta = \sum_{i=1}^K \lambda_{m_i,\mathrm w}^2P_{\max}^2$. {\blue Recall that a Gamma distribution converges to a Gaussian distribution as its shape parameter increases \cite[Sec. 5.5]{Casella02StatInference}, \cite[Sec. 2.8]{Ross19Probability}. In our scenario, the sufficiently large number of cooperative users, $K$, results in a large shape parameter \eqref{eq:test_statistic_RV_GammaApprox}.} 

Consequently, conditioned on the selected set $\mathcal{S}$, the equation of \eqref{eq:test_statistic_RV_GammaApprox} can be approximated to
\begin{equation} \label{eq:test_statistic_RV_CLTApprox_fixedSet}
   T_\mathrm w - \sigma_\mathrm w^2 \mid \mathcal{S} \\[-2pt]
      \sim 
     \begin{cases}
      \displaystyle \mathcal{N} \left( \theta_\mathcal{S}, \Theta_\mathcal{S} \right), 
      \hfill \mathcal{H}_0, \\
      \displaystyle \mathcal{N} \left( \theta_\mathcal{S}, \Theta_\mathcal{S} \right) \ast \Gamma\left(1, \lambda_\mathrm{a,w}P_\mathrm a
      \right), 
      \hfill \mathcal{H}_1,
     \end{cases} 
\end{equation}
where $\theta_{\mathcal{S}}=\sum_{m\in\mathcal{S}}\lambda_{m,\mathrm w}P_{\max}$ and $\Theta_{\mathcal{S}}=\sum_{m\in\mathcal{S}}\lambda_{m,\mathrm w}^{2}P_{\max}^{2}$. 
As $\mathcal{S}$ is randomly drawn and unknown to Willie, the overall statistics of $T_\mathrm w-\sigma_\mathrm w^2$ must incorporate this selection induced randomness. Equivalently, the mean and variance must be computed by averaging the conditional moments with respect to the random user selection set $\mathcal{S}$.

{\blue By the law of total expectation and variance~\cite[Sec. 4.4]{Casella02StatInference}, \cite[Ch. 3]{Ross19Probability},
the overall mean and variance of the aggregate cooperative interference are obtained by averaging over the order-statistics-induced selected user sequence. 
Since \(r_m=|h_{m,\mathrm b}|^2/\lambda_{m,\mathrm w}\) follows an exponential distribution with rate \(\alpha_m=\lambda_{m,\mathrm w}/\lambda_{m,\mathrm b}\), the selected cooperative users are not, in general, uniformly distributed over all \(K\)-user subsets. 
Rather, their selection is induced by the order-statistics of \(\{r_m\}_{m=1}^{M}\). 
For an ordered selected sequence \((m_1,\ldots,m_K)\), where all indices are distinct, the probability that these users are selected as the first \(K\) users in this order is given by
\begin{equation}
    p(m_1,\ldots,m_K)
    =
    \prod_{\ell=1}^{K}
    \frac{\alpha_{m_\ell}}
    {\sum_{j\notin\{m_1,\ldots,m_{\ell-1}\}}\alpha_j},
    \label{eq:ordered_selection_probability}
\end{equation}
where \(\{m_1,\ldots,m_{\ell-1}\}\) is empty for \(\ell=1\).
Then, the first-order inclusion probability of user \(m\) is defined as
\begin{align}
    \pi_m^{(K)}
    &\triangleq
    \sum_{\substack{(m_1,\ldots,m_K):\\
    m\in\{m_1,\ldots,m_K\}}}
    p(m_1,\ldots,m_K),
    \label{eq:first_order_inclusion_probability}
\end{align}
and the second-order inclusion probability of users \(m\) and \(n\), \(m\neq n\), is defined as
\begin{align}
    \pi_{m,n}^{(K)}
    &\triangleq
    \sum_{\substack{(m_1,\ldots,m_K):\\
    m,n\in\{m_1,\ldots,m_K\}}}
    p(m_1,\ldots,m_K).
    \label{eq:second_order_inclusion_probability}
\end{align}
In \eqref{eq:first_order_inclusion_probability} and
\eqref{eq:second_order_inclusion_probability}, the summations are taken over all ordered selected sequences \((m_1,\ldots,m_K)\) with distinct indices.

Then, \(\Xi_K\triangleq P_{\max}\sum_{m=1}^{M}\pi_m^{(K)}\lambda_{m,\mathrm w}\) as stated in Lemma~\ref{lemma:DEP}, and
\(\Sigma_K\triangleq P_{\max}^{2}\left(\mathbb E[A_K]+\mathrm{Var}(L_K)\right)\),
where \( L_K\triangleq \sum_{i=1}^{K}\lambda_{m_i,\mathrm w}\) and \( A_K\triangleq \sum_{i=1}^{K}\lambda_{m_i,\mathrm w}^{2}\).
Expanding \(\mathbb E[A_K]\) and \(\mathrm{Var}(L_K)\) with the inclusion probabilities \(\pi_m^{(K)}\) and \(\pi_{m,n}^{(K)}\) in \eqref{eq:first_order_inclusion_probability} and \eqref{eq:second_order_inclusion_probability} yields exactly \eqref{eq:Sigma_K_selection}.}

Accordingly, the false-alarm probability, $P_{FA}$, is defined as

\begin{equation}\label{eq:PFA_Qform}
P_{\mathrm{FA}}
= \Pr\!\left(T_\mathrm w-\sigma_{\mathrm w}^2>\hat{\gamma}\,\middle|\,\mathcal{H}_0\right)
\approx
Q\!\left(
\frac{\hat{\gamma}-\Xi_K}{\sqrt{\Sigma_K}}
\right),
\end{equation}
where \(\hat{\gamma}\triangleq\gamma-\sigma_{\mathrm w}^2\). In addition, the miss-detection probability at Willie can be derived as
\begin{align}\label{eq:PMD_Qform}
P_{\mathrm{MD}}
&=
\Pr\!\left(T_\mathrm w-\sigma_{\mathrm w}^2\le\hat{\gamma}\,\middle|\,\mathcal{H}_1\right)
=
\Pr(X+Y\le\hat{\gamma})
\nonumber\\
&=
\int_{-\infty}^{\hat{\gamma}}
f_Y(y)
\left[
1-\exp\!\left(
-\frac{\hat{\gamma}-y}{\Delta}
\right)
\right]dy
\nonumber\\
&=
Q\!\left(
-\frac{\hat{\gamma}-\Xi_K}{\sqrt{\Sigma_K}}
\right)
-
\exp\!\left(a_\gamma\right) Q\!\left(b_\gamma\right),
\end{align}
where \(X\sim\Gamma(1,\Delta)\) with \(\Delta\triangleq P_{\mathrm a}\lambda_{\mathrm a,\mathrm w}\), and \(f_Y(y)\) denotes the PDF of \(Y\sim\mathcal{N}(\Xi_K,\Sigma_K)\).
Define
\begin{equation*}
a_\gamma \triangleq -\frac{2\Delta(\hat{\gamma}-\Xi_K)-\Sigma_K}{2\Delta^2},
\qquad
b_\gamma \triangleq -\frac{\Delta(\hat{\gamma}-\Xi_K)-\Sigma_K}{\sqrt{\Sigma_K}\,\Delta}.
\end{equation*}
Combining \eqref{eq:PMD_Qform} with \eqref{eq:PFA_Qform} gives
\(\zeta(\gamma)=P_{\mathrm{FA}}+P_{\mathrm{MD}}=1-\exp(a_\gamma)\,Q(b_\gamma)\),
which matches \eqref{eq:DEP_general}.

\bibliographystyle{IEEEtran}
\bibliography{ref_covert}

\end{document}